\documentclass[11pt]{article}

\usepackage{authblk}
\usepackage{amsmath,amsfonts,amssymb,amsthm}
\usepackage{enumerate}

\usepackage{mathptmx}  
\usepackage{epsfig, psfrag, epstopdf,graphicx}

\let\a=\alpha \let\b=\beta  \let\g=\gamma  \let\d=\delta \let\e=\varepsilon
  \let\h=\eta   \let\th=\theta \let\k=\kappa \let\l=\lambda
\let\m=\mu    \let\n=\nu             
\let\s=\sigma \let\t=\tau    
   \let\o=\omega

\def\EE{{\cal E}} \def\VV{{\cal V}}
 \def\WW{{\cal W}}
\def\TT{{\cal T}} 
\def\RR{{\cal R}}\def\LL{{\cal L}}  
\def\DD{{\cal D}}

 \def\pp{{\bf p}}
 \def\xx{{\bf x}} \def\yy{{\bf y}} \def\zz{{\bf z}}
\def\kk{{\bf k}}
\def\PP{{\bf P}}

\def\nn{\nonumber}

\def\RRR{\hbox{\msytw R}}

 \def\ZZZ{\hbox{\msytw Z}}

\def\\{\hfill\break}
\def\={:=}
\let\io=\infty
\let\0=\noindent

\def\tende#1{\,\vtop{\ialign{##\crcr\rightarrowfill\crcr\noalign{\kern-1pt
    \nointerlineskip} \hskip3.pt${\scriptstyle #1}$\hskip3.pt\crcr}}\,}
\def\otto{\,{\kern-1.truept\leftarrow\kern-5.truept\to\kern-1.truept}\,}

\def\to{\rightarrow}

\def\qed{\hfill\raise1pt\hbox{\vrule height5pt width5pt depth0pt}}

\def\lis{\overline}

\def\be{\begin{equation}}
\def\ee{\end{equation}}
\def\bea{\begin{eqnarray}}
\def\eea{\end{eqnarray}}
\def\nn{\nonumber}
\def\pref#1{(\ref{#1})}

\newtheorem{lemma}{Lemma}[section]
\newtheorem{theorem}{Theorem}[section]

\newcount\driver
\newcount\bozza

scaled\magstep1 \font\msytw=msbm10 scaled\magstep1
 
scaled\magstep1

{\count255=\time\divide\count255 by 60
\xdef\hourmin{\number\count255}
        \multiply\count255 by-60\advance\count255 by\time
   \xdef\hourmin{\hourmin:\ifnum\count255<10 0\fi\the\count255}}

\let\a=\alpha \let\b=\beta    \let\g=\gamma     \let\d=\delta     \let\e=\varepsilon
  \let\h=\eta     \let\th=\vartheta \let\k=\kappa     \let\l=\lambda
\let\m=\mu    \let\n=\nu                      
\let\s=\sigma \let\t=\tau            
   \let\o=\omega

\def\PP{{\cal P}}\def\EE{{\cal E}}\def\VV{{\cal V}}
\def\WW{{\cal W}}
\def\TT{{\cal T}}
\def\RR{{\cal R}}\def\LL{{\cal L}}
\def\DD{{\cal D}}

\def\RRR{\mathbb{R}} \def\ZZZ{\mathbb{Z}}  

       \def\ZZZ{\hbox{\msytw Z}}

\def\pp{{\bf p}}\def\xx{{\bf x}}
 
\def\yy{{\bf y}}\def\kk{{\bf k}}\def\nn{{\bf n}}
\def\zz{{\bf z}}

\def\be#1\ee{\begin{equation}#1\end{equation}}
\def\bsp#1\esp{\begin{split}#1\end{split}}
\def\bal#1\eal{\begin{align}#1\end{align}}
\def\bald#1\eald{\begin{aligned}#1\end{aligned}}
\def\ba#1#2\ea{\begin{array}{#1}#2\end{array}}

\def\bea{\begin{eqnarray}}   \def\eea{\end{eqnarray}}
\def\bean{\begin{eqnarray*}} \def\eean{\end{eqnarray*}}
\def\bfr{\begin{flushright}} \def\efr{\end{flushright}}
\def\bc{\begin{center}}      \def\ec{\end{center}}
\def\bd{\begin{description}} \def\ed{\end{description}}
\def\bv{\begin{verbatim}}

\def\nn{\nonumber}
\def\Halmos{\hfill\vrule height10pt width4pt depth2pt \par\hbox to \hsize{}}
\def\pref#1{(\ref{#1})}

\def\qed{\raise1pt\hbox{\vrule height5pt width5pt depth0pt}}

\let\==\equiv

\let\io=\infty
\let\0=\noindent

\def\*{\vspace{.2cm}}

\def\tilde#1{{\widetilde #1}}

\def\ins#1#2#3{\vbox to0pt{\kern-#2 \hbox{\kern#1 #3}\vss}\nointerlineskip}

\newdimen\xshift \newdimen\xwidth \newdimen\yshift
\newcount\griglia

\def\insertplot#1#2#3#4#5#6{%
\xwidth=#1pt \xshift=\hsize \advance\xshift by-\xwidth \divide\xshift by 2%
\begin{figure}[ht]
\vspace{#2pt} \hspace{\xshift}
\begin{minipage}{#1pt}
#3 \ifnum\driver=1 \griglia=#6
\ifnum\griglia=1 \openout13=griglia.ps \write13{gsave .2
setlinewidth} \write13{0 10 #1 {dup 0 moveto #2 lineto } for}
\write13{0 10 #2 {dup 0 exch moveto #1 exch lineto } for}
\write13{stroke} \write13{.5 setlinewidth} \write13{0 50 #1 {dup 0
moveto #2 lineto } for} \write13{0 50 #2 {dup 0 exch moveto #1
exch lineto } for} \write13{stroke grestore} \closeout13
\includegraphics{griglia.ps} \fi
\includegraphics{#4.ps}\fi%
\ifnum\driver=2 \fi
\end{minipage}
\caption{#5}
\end{figure}
}

\def\cal#1{\mathcal{#1}}

\begin{document}

\title{Quantum Phase Transition in an Interacting Fermionic Chain.}

\author[1]{F. Bonetto} 
\author[2]{V. Mastropietro} 
\affil[1]{School of Mathematics, 
Georgia Institute of
Technology, 
Atlanta, GA 30332 USA\break
bonetto@math.gatech.edu}

\affil[2]{Dipartimento di Matematica ``Federigo Enriques",
Universit\`a degli Studi di Milano,\break
Via Cesare Saldini 50,
Milano, Italy\break
vieri.mastropietro@unimi.it}
\maketitle

\begin{abstract}
We rigorously analyze the quantum phase transition between a metallic and an 
insulating phase  in (non solvable) interacting spin chains or one dimensional 
fermionic systems. In particular, we prove the persistence of Luttinger 
liquid behavior in the presence of an interaction even arbitrarily close to 
the critical 
point, where the Fermi velocity vanishes and the two Fermi points coalesce. The 
analysis is based on two different multiscale analysis; the analysis of 
the first regime provides gain factors which compensate exactly the small 
divisors due to the vanishing Fermi velocity.
\end{abstract}

\section{Introduction}

\subsection{Spin or fermionic chains}

Recently a great deal of attention has been focused on the quantum phase 
transition between a metallic and an insulating phase  in (non solvable) 
interacting spin chains or one dimensional fermionic systems. Beside its 
intrinsic interest, such problem has a paradigmatic character, see {\it e.g.} 
\cite{S,GL}. Interacting fermionic systems are often investigated using 
bosonization \cite{ML}, but such method cannot be used in this case; it 
requires 
linear dispersion relation, while in our case close to the critical point the 
dispersion 
relation becomes quadratic. Interacting fermionic systems with non linear 
dispersion relation have been studied using convergent expansions, based on 
rigorous Renormalization Group methods. However the estimate for the radius 
of 
convergence of the expansions involved vanishes at the critical point, so that 
they provide no information 
close to the quantum phase transition. This paper contains the first rigorous 
study of 
the behavior close to the metal insulator transition, using an expansion 
convergent uniformly in a  region of parameters including the critical point.

We will focus for definiteness on the model whose Hamiltonian is given by
\be
H=-\sum_x \frac12 [S^1_x S^1_{x+1}+S^2_x S^2_{x+1}]-\l \sum_{x,y} v(x-y) S^3_x
S^3_{y}-\bar h \sum_x S^3_x+U_L\label{xx}
\ee
where  $(S^1_x,S^2_x,S^3_x)=\frac{1}{2}(\s^1_x,\s^2_x,\s^3_x)$, for 
$x=1,2,...,L$, $\s^i_x$, $i=1,2,3$ are Pauli matrices, $\bar h$ is the 
magnetic field and $v(x)$ is a short range even potential, that is
\[
 v(-x)=v(x),\qquad |v(x)|\leq e^{-\k x}.
\]
Finally $U_L$ is an 
operator depending only on $S_1^i$ and $S_L^i$ to be used later to fix the 
boundary conditions. When $v(x-y)=\d_{x,y+1}$, this model is known as $XXZ$ 
Heisenberg spin chain. Setting $\xx=(x_0,x)$, we define $S^i_\xx=e^{H x_0}S^i_x 
e^{-H x_0}$. Moreover, given an observable $O$, we define 
\be
\langle O\rangle_{\b,L}=\frac{{\rm Tr}\,e^{-\b H} O}{{\rm Tr}\, e^{-\b H}}
\qquad\qquad\hbox{and}\qquad\qquad
\langle O\rangle=\lim_{\b\to\io}\lim_{L\to\io}\langle O\rangle_{\b,L}.
\ee
It is well known that spin chains can be rewritten in terms of fermionic 
operators $a^\pm_x$, with $\{a^+_{x}, a^-_{y}\}=\d_{x,y}$, $\{a^+_{x}, 
a^+_{y}\}=\{a^-_{x}, a^-_{y}\}=0$, by the {\it Jordan-Wigner transformation}:
\be
\s^-_x=e^{-i\pi \sum_{y=1}^{x-1} a^+_y a^-_y }  a^-_x\;,
\quad \s^+_x= a^+_x e^{i\pi\sum_{y=1}^{x-1} a^+_y a^-_y }\;,\quad
\s^3_x=2 a^+_x a^-_x-1
\ee
where $\s_x^{\pm}=(\s_x^1\pm i\s_x^2)/2$. In terms of the fermionic operators 
the Hamiltonian becomes
\be
H=- \sum_x \left[\frac12(a^+_{x+1}a^-_x+a^+_{x}a^-_{x+1})+h a^+_x
a^-_x\right]-\l\sum_{x,y} v(x-y)
a^+_{x} a^-_x a^+_{y} a^-_{y}\label{sa}
\ee
where $ h=\bar h-\l\hat v(0)$ and $U_L$ can be chosen so to obtain periodic 
boundary conditions for the fermions, {\it i.e.} $a^\pm_L=a^\pm_1$. Therefore 
the spin chain \pref{xx} can be equivalently represented as a model for 
interacting spinless fermions in one dimension with {\it chemical potential} 
$\m=-h$. 

The 2-point Schwinger function is defined as 
\be
S_{L,\b}(\xx-\yy)=\langle {\bf T} a^-_{\xx} a^+_{\yy}\rangle_{L,\b}
\ee
where ${\bf T}$ is the {\it time ordering} operator, that is  ${\bf T}(
a^-_{\xx} 
a^+_{\yy})=a^-_{\xx} a^+_{\yy}$ if $x_0>y_0$ and ${\bf T} (a^-_{\xx} 
a^+_{\yy})=-a^+_\yy a^-_{\xx}$ if $x_0\le y_0$. We will mostly study the 
infinite volume zero temperature 2-points Schwinger function given by 
\[
\lim_{\b\to\io}\lim_{L\to\io}S_{L,\b}(\xx-\yy)=S(\xx-\yy).
\]

\subsection{Quantum Phase transition in the non interacting case}

The fermionic representation makes the analysis of the $\l=0$ case (the so
called $XX$ chain) quite immediate; writing 
\[
a^\pm_{x}=\frac{1}{L}\sum_k e^{\pm ik x }\hat a^\pm_k
\]
the Hamiltonian becomes
\be\label{H}
H_0=\frac{1}{L}\sum_k \e(k)\hat a^+_k \hat a^-_k\quad\quad\quad \e(k)=-\cos k-h
\ee
where $k=\frac{2\pi n}{L}$, $-\pi\le k< \pi$. 

The {\it ground state} of \eqref{H} depends critically on $h$. Indeed, for  
$h<-1$ the ground 
state is the fermionic vacuum (empty band {\it insulating} state), for $h> 1$ it 
is the state with all fermionic levels occupied (filled band {\it insulating} 
state) and $-1<h<1$ the ground state corresponds to the state in which all the 
fermionic levels with momenta $|k|\le p_F=\arccos (-{ h})$, are occupied 
({\it metallic} state). $p_F$ is called {\it Fermi momentum} and $\pm p_F$ are 
the {\it Fermi points} (the analogous of the Fermi surface in one dimension). In 
other words the values $h=\pm 1$ separate two different behaviors at 
zero temperature; one says that  in correspondence of $h=\pm 1$ there is a {\it 
quantum phase transition} \cite{S} between a metallic and an insulating phase. 

The metallic or insulating phases are signaled by different properties of the 
two point Schwinger function, which is given by
\begin{align}\label{S0bL}
S_{0,L,
\b}(\xx-\yy)=\frac{1}{L}\sum_{k} e^{ikx}&\left\{
\frac{e^{-(x_0-y_0)\e(k)}} {1+e^{-\b\e(k)}}\th(x_0-y_0>0)-\right.\crcr
&\quad\left.\frac{e^{-(\b+x_0-y_0)\e(k)}}{1+e^{-\b\e(k)}}\th(x_0-y_0\le 
0)\right\}
\end{align}
where $\vartheta(x_0)=1$ if $x_0>0$ and $\vartheta(x_0)=0$ otherwise. The 
Schwinger function \eqref{S0bL} is defined for $-\b\le x_0\le \b$ but it can be 
extended periodically over the whole real axis; such extension is smooth in 
$x_0-y_0$ for $x_0-y_0\not =n\b$, $n\in \ZZZ$ and it is discontinuous at 
$\xx-\yy=(n\b,0)$. Since  $S_{0,L, \b}(\xx)$ is antiperiodic in $x_0$ its 
Fourier series is of the form
\be
S_{0,L, \b}(\xx)=\frac{1}{L}\sum_{k} e^{ikx}\hat S_{0,L,
\b}(k,x_0)=\frac{1}{\b L}\sum_{\kk\in \mathcal{D}}e^{-i \kk 
\xx}\hat S_{0,\b,L}(\kk)
\ee
with $\kk=(k_0,k)$, ${\cal D}=\left\{\kk\,|\, k=\frac{2\pi 
m}{L},\;-\pi\leq k < \pi,\; k_0=\frac{2\pi}{\b}(n+\frac{1}{2})\right\}$ and
\be
\hat S_{0,L,\b}(\kk)=\frac{1}{-i k_0+\cos k+h}
\ee
In the metallic phase the Schwinger function $\hat S_{0}(\kk)$ is singular in 
correspondence of the Fermi points $(0,\pm p_F)$. For $|k|$ close to $p_F$ we 
have $\hat S_0(\kk)\sim \frac{1}{ -i k_0+v_F(|k|-p_F)}$. Notice that the 2-point 
Schwinger function is asymptotically identical, if the momenta are measured from 
the Fermi points, to the Schwinger function of massless Dirac fermions in 
$d=1+1$ 
with Fermi velocity $v_F$. For values of $h$ close to $h=-1$ ({\it i.e. } for 
small positive $r$ if we set $h=-1+r$) both the distance of the Fermi points and 
$v_F$ are $O(\sqrt{r})$, that is the Fermi velocity vanishes with continuity and 
the two Fermi points coalesce. At criticality when $r=0$ the 2-point 
function $\hat S_0(\kk)$ is singular only at $(0,0)$ and $\hat S_0(\kk)\sim 
\frac{1}{ -i k_0+\frac12k^2}$; the elementary excitations do not have a 
relativistic linear dispersion relation, as in the metallic phase, but a 
parabolic one. Finally  in the insulating phase for $r<0$ the two point function 
has no singularities.

It is natural to ask what happens to the quantum phase transition in presence of 
the interaction.

\subsection{Quantum Phase transition in the interacting case}

The Schwinger functions of the interacting model in the metallic phase have been 
constructed using Renormalization Group (RG) methods in \cite{BM1, BM2, BM3}. 
Luttinger liquid behavior (in the sense of \cite{Ha}) has been established, 
showing that the power law decay of correlations is modified by the 
interaction via the appearance of critical exponents, that depend in a non 
trivial way on the interaction. It should be stressed that such analysis 
provides a full understanding inside the metallic phase, but gives no 
information on the phase transition; the reason is that the physical observables 
are expressed in terms of renormalized expansions which are convergent 
under the condition
\be
|\l|\le \e |v_F|
\ee
and small $\e$; therefore, the closer one is to the bottom (or the top) of the 
band, the smaller the interaction has to be chosen. This is not surprising, as 
such RG methods essentially show that the interacting fermionic chain is 
asymptotic to a system of interacting massless Dirac fermions in $d=1+1$ 
dimensions with coupling $\frac{\l}{v_F}$. One may even suspect that an 
extremely  weak interaction could produce some quantum instability close to the 
boundary of the metallic phase, where the parameters correspond to a strong 
coupling regime in the effective description.

This is however excluded by our results; we prove the persistence of the 
metallic phase, with Luttinger liquid behavior, in presence of interaction even 
arbitrarily close to the critical point, where the Fermi velocity vanishes. This 
result is achieved writing the correlations in terms of a renormalized expansion 
with a radius of convergence which is independent from the Fermi velocity. In 
order to obtain this result we needs to exploit the non linear corrections to 
the dispersion relation due to the lattice. The proof is indeed based on two 
different multiscale analysis  in two regions of the energy momentum space; in 
the smaller energy region the effective relativistic description is valid while 
in the larger energy region the quadratic corrections due to the lattice are 
dominating. The scaling dimensions in the two regimes are different; after the 
integration of the first regime one gets gain factors which compensate exactly 
the velocities at the denominator produced in the second regime, so that 
uniformity is achieved.

Our main results are summarized by the following theorem. We state it in terms 
of the Fourier transform of the 2-points Schwinger function defined by
\be
S_{L,\beta}(\xx)=\frac{1}{\beta L}\sum_{\kk\in{\cal D}} e^{i\kk\xx}\hat 
S_{L,\beta}(\kk)\qquad \hat 
S(\kk)=\lim_{\b\to\infty}\lim_{L\to\infty}\hat S_{L,\beta}(\kk).
\ee
\vskip.3cm
\begin{theorem} Given $h=-1+r$ with $|r|<1$, there exists $\e>0$ and $C>0$ 
(independent form $L,\b,r$) such that, for $|\lambda|<\e$, we have: 
\begin{enumerate}
\item For $r> 0$ (metallic phase), 
\be
\hat S_{L,\beta}(\kk\pm p_F)= \frac{[k_0^2+\a(\l)^2(\cos 
k-1+\n(\l))^2]^{\h(\l)}}{ 
-i k_0+\a(\l)(\cos k-1+\n(\l))}(1+ \l R_S(\l,\kk))
\label{ffoop}
\ee
where 
\begin{align}
\n(\l)=&r+\l rR_\n(\l)\qquad\a(\l)=1+\l R_\a(\l)\crcr
\h(\l)=&b \l^2r+\l^3 r^{\frac32} R_\h(\l)
\end{align}
with $b>0$ a constant and $|R_i|\leq C$ for $i=S,\n,\a$ and $\h$.
\item For $r=0$ (critical point)
\be
\hat S_{L,\beta}(\kk)= \frac{1+\l R_S(\l,\kk)}{ -i k_0+ \a(\l) 
(\cos(k)-1)}\label{mama}
\ee
where $\a(\l)=1+\l R_\a(\l)$ and $|R_i|\leq C$ for $i=\a,S$.
\item For $r<0$ (insulating phase)  
\be
|\hat S_{L,\beta}(\kk)|\le \frac{C}{ |r|}
\ee
\end{enumerate}
\end{theorem}

Clearly, by symmetry, similar results hold at the top of the band by setting 
$h=1-r$.

From the above result we see that in the metallic phase Luttinger liquid 
behavior is present; indeed the interaction changes the location of the Fermi 
points from $p_F=\pm \cos^{-1}(-1+r)$ to $p_F=\cos^{-1}(-1+r+O(\l r))$ and, more 
remarkably, produces an anomalous behavior in the two point Schwinger function 
due to the presence of the critical exponent $\h$. Luttinger liquid behavior 
persists up to the critical point (corresponding to a strong coupling phase in 
the effective relativistic description); interestingly, the critical exponent 
becomes smaller the closer one is to the critical point.  This is due to the 
fact that the effective coupling is $O(\l r)$ (and not $O(\l)$), so that the 
effective coupling divided by the Fermi velocity is $O(\sqrt{r})$ and indeed  
small for small $r$

At the critical point no anomalous exponent is present; the asymptotic 
behavior is qualitatively the same as in the non interacting case, up to a 
finite wave function renormalization and the presence of $\a$. Finally in the 
$r<0$ again an insulating behavior is found, as the 2-point function has no 
singularities.

The proof of the above result is based on a rigorous implementation of the 
Wilsonian RG methods. There is a natural momentum scale, which is $O(r)$,  
separating two different regimes. In the first regime, described in section 2,  
the interaction appears to be  relevant; however, Pauli principle shows that the 
relevant contributions are vanishing and the theory turns out to be effectively 
superrenormalizable: all the interactions are irrelevant and their effect is to 
produce a finite renormalization of the parameters. In the insulating phase or 
at the critical point, only this regime is present. In contrast, in the metallic 
phase there is a second regime, described in section 3,  in which the relevant 
contribution are non vanishing (the presence of the Fermi points introduces an 
extra label in the fermionic fields). The local quartic terms are therefore 
marginal and produce the critical exponents; in this second regime, one has to 
check carefully that the small divisors due to the fact that the Fermi velocity 
is small are compensated by the small factors coming from the integration of the 
first regime.

\section{Renormalization Group integration: the first regime}

\subsection{Grassmann representation}

We introduce a set of {\it anticommuting variables} $\psi^\pm_\xx$
such that $\{\psi^\pm_\xx,\psi^\pm_\yy\}=0$ for every $\xx,\yy$. Given the
propagator
\be
g_{M,L,\beta}(\xx-\yy)=\frac{1}{ \b L}\sum_{\kk\in\DD}
e^{i\kk(\xx-\yy)}\frac{\chi_0(\g^{-M}|k_0|)}{ -i k_0+\cos k+h}\label{pro}
\ee
where $\chi_0(t)$ is a smooth even compact support function equal to $1$ 
for $|t|\le 1$ and equal to $0$ for $|t|\ge \g$, for some $1<\g\leq 2$,
we define the {\it Grasmann integration} $P(d\psi)$ on the {\it
Grassman algebra} generated by the $\psi^\pm_\xx$ by setting
\[
 \int P(d\psi)\,\prod_{i=1}^n\psi^-_{\xx_i}\,\prod_{j=1}^n\psi^+_{\yy_j}=\det G
\]
where $G$ is the $n\times n$ matrix with entries 
$G_{i,j}=g_{M,L,\beta}(\xx_i-\yy_j)$. We can extend this definition to a generic 
monomial in the $\psi^\pm$ using the anticommutation rule and to the full 
algebra by linearity.

Observe that 
if $\xx\not= (0,n\b)$
\be
\lim_{M\to\io} g_{M,L,\beta}(\xx)=S_{0,L,\beta}(\xx)\label{pro1}\ee
while for $\xx=(0,n\b)$
\be
\lim_{M\to\io} 
g_{M,L,\beta}(\xx)=\frac{S_{0,L,\beta}(0^+,0)+S_{0,L,\beta}(0^-,0)}{ 
2}\label{sasas}
\ee
while $S_{0,L,\beta}(n\b, 0)=S_{0,L,\beta}(0^-,0)$.

By extending the Grassmann algebra with a new set of anticommuting
variables $\phi^\pm_\xx$, we can define the {\it Generating Functional} 
$\WW(\phi)$ as the following Grassmann integral
\be
e^{\WW_M(\phi)}=\int P(d\psi)e^{\VV(\psi)+(\psi,\phi)}\label{fi}
\ee
where, if $\int d\xx=\int dx_0\sum_x$, we set
\begin{align}
(\psi,\phi)&=\int d\xx [\psi^+_\xx\phi^-_\xx+\psi^-_\xx\phi^+_\xx]
\label{V}\\
\VV(\psi)&=\l \int d\xx d\yy v(\xx-\yy) \psi^+_\xx\psi^-_\xx
\psi^+_\yy\psi^-_\yy+\bar\n\int d\xx \psi^+_\xx
\psi^-_\xx
\end{align}
and $v(\xx-\yy)=\d(x_0-y_0)v(x-y)$; moreover 
\be
\bar\n=\l \hat v(0)\left[
\frac{S_{0,L,\beta}(0,0^+)-S_{0,L,\beta}(0,0^-)}{ 2}\right].
\ee
The presence in \pref{fi} of the counterterm $\bar\n$ is necessary to take into
account the 
difference between $g(\xx)$ and $S_0(\xx)$, see \pref{sasas}.

We define the Schwinger functions as
\be
S_{M,L,\b}(\xx-\yy)=\frac{\partial^2}{\partial\phi^+_{\xx}\partial\phi^-_\yy}
W(\phi)|_{\phi=0}
\ee
One can easily check, see for instance Proposition 2.1 of \cite{BFM1}, that if 
$S_{\b,L,M}$ is analytic for $|\l|\leq\e$ then $\lim_{M\to\io} 
S_{M,L,\b}(\xx)$ coincides with $S_{L,\b}(\xx)$ for $\l\leq\e$. Thus if we can 
show that $\e$ does not depend on $L,\beta$ and that the convergence of  
$S_{M,L,\b}(\xx)$ as to $S_{L,\b}(\xx)$ is uniform in $L,\beta$, we can 
study the two-point function of \pref{sa} by analyzing the Generating 
functional \pref{fi}.

For definiteness, we take $|r|\leq 1/2$. The remaining range of $r$ is covered 
by the results in \cite{BM1}. The starting point of the analysis is the 
following decomposition of the
propagator
\be
g_{M,L,\beta}(\xx-\yy)=g^{(> 0)}(\xx-\yy)+g^{(\leq 0)}(\xx-\yy)
\ee
where 
\be
g^{(\leq 0)}(\xx-\yy)=\int d\kk
e^{i\kk(\xx-\yy)}\frac{\chi_0(\g^{-M}|k_0|)\chi_{\leq 0}(\kk)}{ -i k_0+\cos
k+h}\label{prol0}
\ee
where $\int d\kk$ stands for $\frac{1}{ \b 
L}\sum_{\kk\in\DD}$, $\chi_{<0}(\kk)=\chi_0\left(a_0^{-1}\sqrt{k_0^2+(\cos 
k-1+r)^2}\right)$, and\\
$g^{(> 0)}(\xx-\yy)$ is equal to \pref{prol0}
with $\chi_{\leq 0}(\kk)$ replaced by $(1-\chi_{\leq 0}(\kk))$. We chose 
$a_0=\g^{-1}(1/2-r)$ so that, in the support of $\chi_{<0}(\kk)$ we have 
$|k|\leq\pi/6$. This assures that on the domain of $\chi_{<0}$ we have
\[
 c|k|\leq |\sin(k)|\leq C|k|
\]
for suitable constant $c$ and $C$.

By using the addition property of Grassmann integrations we can write
\be
e^{\WW(\phi)}=\int P(d\psi^{(> 0)}) P(d\psi^{(\le 0)})
e^{\VV(\psi^{(> 0)}+\psi^{(\le 0)}
)+(\psi^{(> 0)}+\psi^{(\le 0)}
,\phi)}.
\ee
After integrating the field $\psi^{(> 0)}$ one obtains
\be
e^{\WW(\phi)}=e^{-\b L F_0}
\int P(d\psi^{(\le 0)})
e^{\VV^{(0)}(\psi^{(\le 0},\phi)
}\label{ss}
\ee
It is known, see for instance Lemma 2.2 of \cite{BFM1} for a proof, that
$\VV^{(0)}(\psi^{(\le 0},\phi)$ is given by 
\be 
\VV^{(0)}(\psi,\phi)=\sum_{n,m\ge 0}
\int d\underline\xx \int d\underline\yy
\,\prod_{i=1}^n\psi^{\e_i}_{\xx_i}\,\prod_{j=1}^m\phi^{\s_j}_{\xx_j}\,
W_{n,m}(\underline\xx,\underline\yy)
\ee
where $\underline\xx=(\xx_1,\ldots,\xx_n)$ and
$\underline\yy=(\yy_1,\ldots,\yy_m)$ while $\prod_{i=1}^n\psi^{\e_i}_{\xx_i}=1$
if $n=0$ and \\$\prod_{j=1}^m\phi^{\s_j}_{\yy_j}=1$ if $m=0$; moreover 
$W_{n,m}(\underline\xx,\underline\yy)$ are given by convergent power series in 
$\l$ for $\l$ small enough and they decay faster than any power in any 
coordinate difference. Finally, the limit $M\to\io$ of 
$\VV^{(0)}(\psi,\phi)$ exists and is reached 
uniformly in $\b,L$.

\subsection{The infrared integration}

Thus we are left with the integration over $\psi^{(\le 0)}$. The heuristic 
idea to perform this integration is to decompose $\psi_\xx^{(\le 0)}$ as
\[
\psi_\xx^{(\le 0)}=\sum_{h=0}^{-\infty}\psi_\xx^{(h)}
\]
where $\psi_\xx^{(h)}$ depends only on the momenta $\kk$ such that ${ -i 
k_0+\cos k-1+r}\simeq \g^h$. By using repeatedly the addition property for 
Grasmann integration this decomposition should allow us to integrate 
recursively 
over the $\psi^{(h)}$. The index $h$ is called the {\it scale} of the field 
$\psi^{(h)}$. Two different regimes will naturally appear in the analysis, 
separated by an energy scale depending on $r$. In this section we describe 
in detail the integration over the first scale and then we give the recursive 
procedure. To simplify notation we study only the case $\phi=0$. The 
general case can be obtained easily. 

We saw that after the ultraviolet integration we have
\be
e^{\WW(0)}=e^{-\b L F_0}\int P(d\psi^{(\le h)}) e^{-\VV^{(0)}(\psi^{(\le 
0)})}\label{jj0}
\ee
where $\VV^{(0)}(\psi^{(\le 0)})$ is the {\it effective potential} on scale 0 
and can be written has
\be 
\VV^{(0)}(\psi)=\sum_{n\ge 1}
W^{(0)}_{2n}(\underline\xx,\underline\yy)\prod_{i=1}^{n}\psi^{+}_{\xx_i}
\psi^-_{\yy_i}=\sum_{n\ge 1} \VV_{2n}^{(0)}(\psi) 
\ee
A direct perturbative analysis suggest that to perform the integration 
\eqref{jj0} we need a {\it renormalized} multiscale integration procedure. In 
particular, the terms with $n=1,2$ are {\it relevant} and the terms with $n=3$ 
are {\it marginal}. For this reason we introduce a {\it localization operator} 
acting on the effective potential as
\be \VV^{(0)}=\LL_1 \VV^{(0)}+\RR_1 \VV^{(0)}\label{loc} 
\ee
with $\RR_1=1-\LL_1$ and $\RR_1$ is defined in the following way; 
\begin{enumerate}
\item
$\RR_1 \VV^{(0)}_{2n}=\VV^{(0)}_{2n}$ for $n\ge 4$;
\item
for $n=3, 2$
\be \label{R2}
\RR_1 \VV_4^{(0)}(\psi)=\int \prod_{i=1}^4 d\xx_1 
W_4^{(0)}(\underline 
\xx)
\psi^{+}_{\xx_1}D^+_{\xx_2,\xx_1}
\psi^{-}_{\xx_3}D^-_{\xx_4,\xx_3} 
\ee
\be \label{R3}
\RR_1 \VV_6^{(0)}(\psi)=\int \prod_{i=1}^6 d\xx_i
W_6^{(0)}(\underline \xx)
\psi^{+}_{\xx_1}D^+_{\xx_2,\xx_1}D^+_{\xx_3,\xx_1}
\psi^{-}_{\xx_4}D^-_{\xx_5,\xx_4}D^-_{\xx_6,\xx_4} 
\ee
where
\be
D^\e_{\xx_2,\xx_1}=\psi^\e_{\xx_2}-\psi^\e_{\xx_1}=(\xx_2-\xx_1)\int_0^1
dt {\bf \partial}\psi^\e_{\xx'_{1,2}(t)}\label{D}
\ee
with $\xx'_{1,2}(t)=\xx_1+t(\xx_1-\xx_2)$ will be called an {\it interpolated 
point}.
\item For $n=1$
\be
\RR_1 \VV_2^{(0)}(\psi)=\int d\xx_1 d\xx_2 W_2^{(0)}(\underline
\xx)\psi^{+}_{\xx_1}H^-_{\xx_1,\xx_2}
\ee
where
\begin{align}
H^-_{\xx_1,\xx_2}=&
\psi^-_{\xx_2}-\psi^-_{\xx_1}-(x_{0,1}-x_{0,2})\partial_0\psi^-_{\xx_1}-(x_{1}
-x_{2})\partial_1\psi^-_{\xx_1}-\crcr
&\frac{1}{ 2}(x_1-x_2)^2\tilde\Delta_1\psi
\end{align}
and
\begin{align*}
\tilde\partial_1\psi^-_{\xx}&=\frac 12(\psi^-_{\xx+(0,1)}-\psi^-_{\xx-(0,1)})
=\int d\kk \sin k e^{i\kk\xx}\hat\psi^-_\kk\crcr
\tilde\Delta_1\psi^-_{\xx}&=\psi^-_{\xx+(0,1)}-2\psi^-_{\xx}+\psi^-_{\xx-(0,1)}
=2\int
d\kk (\cos k-1) e^{i\kk\xx}\hat\psi^-_\kk 
\end{align*}
\end{enumerate}

As a consequence of the above definitions
\begin{align}
 \LL_1\VV^{(0)}=&\hat W_2^{(0)}(0) \int d\xx\psi^+_\xx\psi^-_\xx+
\partial_0 \hat W_2^{(0)}(0)
\int d\xx 
\psi^+_\xx\partial_0\psi_\xx^-+\crcr
&{\frac12}\partial^2_1 \hat W_2^{(0)}(0)
 \int d\xx
\psi^+_\xx\partial^2_1\psi_\xx^-\label{blue} 
\end{align}
where we have used that
\begin{enumerate}[i.]
\item $g^{(0)}(k_0,k)=g^{(0)}(k_0,-k)$, so that we get
\be
\partial_1 \hat W^{(0)}_2(0)=0
\ee
\item  There are no terms in $\LL_1\VV^{(0)}$ with four or six fermionic 
fields, 
as
\be
\psi^{\e}_{\xx_1}D^\e_{\xx_2,\xx_1}=\psi^{\e}_{\xx_1}\psi^{\e}_{\xx_2} 
\label{p}.
\ee
and therefore $\RR_1\VV^{(0)}_4=\VV^{(0)}_4$ and 
$\RR_1\VV^{(0)}_6=\VV^{(0)}_6$. As 
a consequence \eqref{R2}\eqref{R3} just represent a different way to 
write the four and six field contribution to the effective potential. This 
representation will be useful in the following where we will exploit the 
dimensional gain due to the zero term $\xx_2-\xx_1$ and the derivative in 
eq.\eqref{D}.

\end{enumerate}

We will call $\LL_1\VV^{(h)}$ the {\it relevant part of the effective 
potential}.
Since it is quadratic in the fields, we can include it in the free 
integration finding
\be
e^{\WW(0)}=e^{-\b L(F_0+e_0)}\int \tilde P(d\psi^{(\le 0)}) e^{
-\RR_1 \VV^{(0)}(\psi^{(\le 0)})}
\label{2.40b}
\ee
where the propagator of $\tilde P(d\psi^{(\le 0)})$ is now
\be
\tilde g^{(\le 0)}(\xx)=\int d\kk
e^{i\kk\xx}\frac{\chi_{\le 0}(\kk)}{ -i 
k_0(1+z_{-1})+(1+\a_{-1})(\cos k-1)+r+\g^{-1}\m_{-1}}
\ee
and
\bea
&&z_{-1}=z_0+\chi_{\le 0}(\kk)\partial_0 \hat W^{(0)}_2(0)\quad\quad 
\a_{-1}=\a_0+\chi_{\le 0}(\kk)\partial^2_1 \hat W^{(0)}_2(0)\nn\\
&&\m_{h-1}=\m_0+\chi_{\le 0}(\kk)\g^{-0}\hat W^{(0)}_2(0)\label{bbe}
\eea
where $z_0=\a_0=\m_0=0$ but we have added them in \eqref{bbe} for later 
reference.

We can now write
\be
\tilde g^{(\leq 0)}(\xx)=g^{(\leq -1)}(\xx)+\tilde g^{(0)}(\xx)\label{dec0}
\ee
where
\be
g^{(\leq -1)}(\xx)=\int d\kk
e^{i\kk\xx}\frac{\chi_{\leq -1}(\kk)}{ -i 
k_0(1+z_{-1})+(1+\a_{-1})(\cos k-1)+r+\g^{-1}\m_{-1}}\label{prol1}
\ee
with 
\[
\chi_{<-1}(\kk)=\chi_0\left(\g a_0^{-1}\sqrt{(1+z_{-1})^2k_0^2+((1+\a_{-1})\cos 
k-1+r+\g^{-1}\m_{-1})^2}\right). 
\]
Clearly 
\be
\tilde g^{(0)}(\xx)=\int d\kk
e^{i\kk\xx}\frac{{f_0}(\kk)}{ -i 
k_0(1+z_{-1})+(1+\a_{-1})(\cos k-1)+r+\g^{-1}\m_{-1}}
\ee
where
\[
 f_{0}(\kk)=\chi_{\leq 0}(\kk)-\chi_{\leq -1}(\kk).
\]
Using again the addition property for Grassmann integrations we 
can rewrite \pref{jj0} and perform the integration over $\psi^{(0)}$ as
\begin{align}
e^{\WW(0)}=&e^{-\b L(F_0+e_0)}\int P(d\psi^{(\le -1)})\int \tilde P(d\psi^{(0)})
 e^{-\RR_1 \VV^{(0)}(\psi^{(\le 0)})}=\\
 =&e^{-\b LF_{-1}}\int P(d\psi^{(\le 
-1)})e^{-\VV^{(-1)}(\psi^{(\le -1)})}
\label{2.40c}
\end{align}
where $\tilde P(d\psi^{(0)})$ is the integration with propagator $\tilde 
g^{(0)}(\xx)$, $P(d\psi^{(\le 1)})$ is the integration with propagator 
$\tilde g^{(\le 1)}(\xx)$ and
\be
 e^{-\b Le_0-\VV^{(-1)}(\psi^{(\le -1)})}=\int \tilde P(d\psi^{(0)})
 e^{-\RR_1 \VV^{(0)}(\psi^{(\le 0)})}
\ee
The fact that this integration is well defined follows from the properties of 
the propagator $\tilde g^{(0)}(\xx)$ that will be derived in Lemma \ref{L1} 
below.

We can now repeat the above procedure iteratively. At the $h$ step (i.e. at 
scale $h$) we start with the integration 
\be
e^{\WW(0)}=e^{-\b LF_{h}}\int P(d\psi^{(\le h)})e^{-\VV^{(h)}(\psi^{(\le 
h)})}\label{jj}
\ee
defined by the propagator

\be
g^{(\leq h)}(\xx)=\int d\kk
e^{i\kk\xx}\frac{\chi_{\leq h}(\kk)}{ -i 
k_0(1+z_{h})+(1+\a_{h})(\cos k-1)+r+\g^{h}\m_{h}}\label{prol1p}
\ee
with 
\[
\chi_{\leq h}(\kk)=\chi_0\left(\g^{-h} 
a_0^{-1}\sqrt{(1+z_{h})^2k_0^2+((1+\a_{h})(\cos 
k-1)+r+\g^{h}\m_{h})^2}\right). 
\]
and the effective potential on scale $h$ is given by
\[ 
\VV^{(h)}(\psi)=\sum_{n\ge 1}
W^{(h)}_{2n}(\underline\xx,\underline\yy)\prod_{i=1}^{n}\psi^{+}_{\xx_i}
\psi^-_{\yy_i}=\sum_{n\ge 1} V_{2n}^{(h)}(\psi) 
\]
Again we can apply the operator $\LL_1$ to $\VV^{(h)}$ to get
\be 
{\cal V}^{(h)}=\LL_1 {\cal V}^{(h)}+\RR_1 {\cal V}^{(h)}\label{loch}\, . 
\ee
where $\RR$ is defined exactly as in the case of $\VV^{(0)}$ and
\begin{align}
\LL_1\VV^{(h)}=&\hat W_2^{(h)}(0) \int d\xx\psi^+_\xx\psi^-_\xx+
\partial_0 \hat W_2^{(h)}(0)
\int d\xx 
\psi^+_\xx\partial_0\psi_\xx^-+\crcr
&{\frac12}\partial^2_1 \hat W_2^{(h)}(0)
 \int d\xx
\psi^+_\xx\partial^2_1\psi_\xx^-\label{bluep} 
\end{align}
Moving the local part of the effective potential into the integration we get
\be
e^{\WW(0)}=e^{-\b L(F_h+e_h)}\int \tilde P(d\psi^{(\le h)}) e^{
-\RR V^{(h)}(\psi^{(\le h)})}
\label{2.40bp}
\ee
where the propagator of $\tilde P(d\psi^{(\le h)})$ is
\be
\tilde g^{(\le h)}(\xx)=\int d\kk
e^{i\kk\xx}\frac{\chi_{\le h}(\kk)}{ -i
k_0(1+z_{h-1})+(1+\a_{h-1})(\cos k-1)+r+\g^{h-1}\m_{h-1}}
\ee
and the {\it running coupling constants} are defined recursively by
\bea
&&z_{h-1}=z_h+\chi_{\le h}(\kk)\partial_0 \hat W^{(h)}_2(0)\quad\quad 
\a_{h-1}=\a_h+\chi_{\le h}(\kk)\partial^2_1 \hat W^{(h)}_2(0)\nn\\
&&\m_{h-1}=\m_h+\chi_{\le h}(\kk)\g^{-h}\hat W^{(h)}_2(0)\label{bbep}
\eea
Finally we can rewrite \pref{2.40b} as
\be
e^{\WW(0)}=e^{-\b L(F_h+e_h)}\int P(d\psi^{(\le h-1)})\int \tilde P(d\psi^{(h)})
 e^{
-\RR_1 V^{(h)}(\psi^{(\le h)})}
\label{2.40cp}
\ee
where $\tilde P(d\psi^{(h)})$ has now propagator
\begin{align}
\tilde g^{(h)}(\xx)=&\int d\kk e^{i\kk\xx}\frac{f_h(\kk)}{ -i 
k_0(1+z_{h-1})+(1+\a_{h-1})(\cos k-1)+r+\g^{h^*}\m_{h-1} }=\crcr
&\int d\kk e^{i\kk\xx}\hat g^{(h)}(\kk)
\label{prop}
\end{align}
and  $f_h(\kk)=\chi_h(\kk)-\chi_{h-1}(\kk)$; one can perform the integration over $\psi^{(h)}$ 
\be
e^{-\b L\bar e_h-{\cal V}^{h-1}}=\int \tilde P(d\psi^{(h)})
 e^{
-\RR_1 {\cal V}^{(h)}(\psi^{(\le h)})}
\label{2.40cc}
\ee
obtaining an expression identical to \pref{jj} with $h-1$ replacing $h$,
so that the procedure can be iterated. 

To show that the above procedure is well defined we need to study the 
propagator $\tilde g^{(\le h)}(\xx)$. We first have to distinguish two range of 
scales. Do do this we set
\[
 h^*=\inf\{h\,|\, a_0\g^{h+1}>|r|\}.
\]
The construction of the 
theory for $r>0$, is based on the fact that the behavior of the propagator 
changes significantly when one reaches the scales $h\simeq h^*$. 
To understand this phenomenon, let's, for simplicity sake, neglect the 
presence of the running constant in the function $\chi_{\leq h}$.  We will 
see in Lemma \ref{L1} and Lemma \ref{L3} below that the presence of 
$\alpha_h$, $z_h$ and $\mu_h$ will not change the picture.
In this situation, it is easy to see that if $h>h^*$ then the domain 
of $f_h(\kk)$ is a ring of width $\g^h$ that goes around both Fermi points 
$(0,\pm p_F)$. At this momentum scale the propagator does not distinguish
between $p_F$ and $-p_F$. 
On the other hand, when $h<h^*$ we have
\be
 k_0^2+(\cos k-1+r)^2>a^2_0\g^{2h+1}\label{cha}
\ee
in an open neighbor of the $k_0$ axis. This means that the domain of $f_h(\kk)$ 
splits in two rings, one around $p_F$ and the other around $-p_F$. In this 
situation it is convenient to write the propagator as a sum of two {\it 
quasi-particle} propagators, each of which depends only on the momenta close 
to one of the Fermi points.

Here we need precise estimates on $\tilde g^{(h)}$ for $h>h^*$ as reported in 
the following Lemma. The case $h\leq h^*$ will be studied in section 3.

\begin{lemma}\label{L1}
Assume that there exists a constant $K>0$ such that 
\be\label{hh}
 |z_h|, |\a_h|, |\mu_h| < K|\l|
\ee
for $h\ge h^*$. Then for every $N$ and $\l$ small enough we have
\be\label{b}
\left|\partial_0^{n_0} \tilde\partial_1^{n_1}\tilde g^{(h)}(\xx)\right|\le
C_N\frac{\g^{\frac{h}{2}}}{1+[\g^h|x_0|+\g^{\frac{h}{2}}|x|]^N}\g^{
h(n_0+n_1/2)}
\ee
with $C_N$ independent form $K$.
\end{lemma}

\begin{proof}
We start observing that the, in the support of $f_h(\kk)$ we have
\be
(1+z_h)|k_0|\le \g^{h+1}a_0\qquad\qquad |(1+\a_{h})(\cos 
k-1)+r+\g^{h}\m_{h}|\le \g^{h+1}a_0\label{hhe}
\ee
From this we get that
\be\label{sin}
|\sin(k)|\le 2\sqrt{1-\cos(k)}\leq 
2\sqrt{\frac{r+K\g^h|\l|+a_0 \g^{h+1}}{1-K|\l|}}\leq C\g^{\frac{h}{2}}.
\ee
so that also $|k|\leq C\g^{\frac{h}{2}}$. It follows that
\be\label{support}
 \int f_{h}(\kk) d\kk \leq C\g^{\frac{3}{2}h}
\ee
To prove the statement for $n_0=n_1=0$, that we just need to show 
that 
\be\label{esti}
|x_0^{N_0}x^{N_1}\tilde 
g^{(h)}(\xx)|\le C
\g^{h\left(\frac{1}{2}-N_0-\frac{N_1}{2}\right)}
\ee
We will use that
\be
 x_0^{N_0}x^{N_1}\tilde g^{(h)}(\xx)=\int d\kk 
e^{i\kk\xx}\partial_0^{N_0}\partial_1^{N_1} \hat g^{(h)}(\kk).
\ee
To estimate the above derivatives we observe that
\be\label{deriv}
\partial_0^{N_0}\partial_1^{N_1} \hat 
g^{(h)}(\kk)=\sum_{P_1=1}^{N_1}\sum_{\sum_ip_i=N_1-P_1}A_{P_1,p_i}\partial_0^{N_0}
\partial^ { P_1 } _ { \cos(k) }\hat g^{(h)}(\kk)
\prod_{i=i}^{P_1}\frac{d^{p_1}}{dk^{p_i}}\sin(k)
\ee
where $A_{P_1,p_i}$ are combinatoric coefficient. It is easy to see that, on 
the domain of $f_h$, we have
\be\label{scaling}
 \left|\partial_0^{N_0}\partial^ { P_1 } _ { \cos(k) }\hat g^{(h)}(\kk)\right|\leq 
C\g^{-h(1+N_0+P_1)}.
\ee
If $P_1\le N_1/2$ we can use
\be\label{prods}
 \left|\prod_{i=i}^{P_1}\frac{d^{p_1}}{dk^{p_i}}\sin(k)\right|\le 1\qquad\hbox{ 
and }\qquad  
\g^{-h(1+N_0+P_1)}\leq \g^{-h\left(1+N_0+\frac{N_1}{2}\right)}
\ee
while, if $P_1> N_1/2$, at least $2P_1-N_1$ of the $p_i$ in the above
product must be zero so that
\begin{align}
\left|\prod_{i=i}^{P_1}\frac{d^{p_1}}{dk^{p_i}}\sin(k)\right|\le& 
C\g^{(2P_1-N_1)h}\crcr  
\g^{-h(1+N_0+P_1)}\g^{h(2P_1-N_1)}\leq& \g^{-h\left(1+N_0+\frac{N_1}{2}\right)}
\end{align}
In both cases we get
\be
|\partial_0^{N_0}\partial_1^{N_1} \hat g^{(h)}(\kk)|\leq 
C\g^{h\left(1+N_0+\frac{N_1}{2}\right)}.
\ee
Combining with \eqref{support} we get \eqref{esti}. Observe now that
\be
 \partial_0^{n_0}\partial_1^{n_1}\tilde g^{(h)}(\xx)=(i)^{n_0+n_1}\int d\kk 
e^{i\kk\xx}k_0^{n_0}\sin^{n_1}k \hat g^{(h)}(\kk).\label{powerk}
\ee
The Lemma follows easily reasoning as above and using \eqref{hhe} for the extra 
powers of $k_0$ and $k$.
\hfill\end{proof}

\subsection{Tree expansion for the effective potentials.}
The effective potential 
$ V^{(h)}(\psi^{(\le h)})$ can be written in terms of a {\it tree expansion}, 
see \cite{G},\cite{BGPS}, defined as follows.

\begin{figure}[ht]
\centering
\includegraphics[width=0.7 \linewidth]{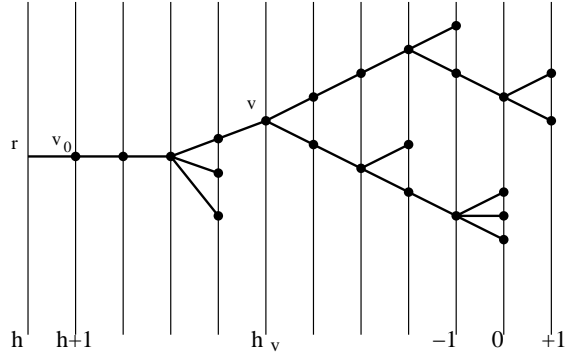}
\vspace{-0.7truecm}
\caption{\label{f1}A tree $\t\in\TT_{h,n}$ with its scale labels.}
\end{figure}

\begin{enumerate}
\item On the plane, we draw the vertical lines at horizontal 
position given by 
the integers from $h$ to 1, see Fig. \ref{f1}. We select one point on the line 
at $h$ ({\it the 
root}) and one point on the line at $h+1$ (the first vertex $v_0$). On the line 
at $k$, with $h+1<k\le 1$, we select $m_k>0$ points (the {\it vertex at scale 
$k$}). We call $M_k$ the set of vertex at scale $k$. To each vertex $v$ in 
$M_k$ we associate exactly one vertex $v'$ in $M_{k-1}$ and we draw a line 
between these two vertices. The vertex $v'$ is called the {\it predecessor} of 
$v$. Finally we require that if $v$ and $w$ are in $M_k$ with $v$ below $w$ 
then $v'$ is below or equal to $w'$. The final results of this procedure is 
clearly a tree with root $r$. 

\item Given a vertex $v$ on scale $k$, let $s_v$ be the number of vertex on 
scale $h+1$ linked to $v$. If $s_v=0$ we say that $v$ is a {\it end point}. The 
number $n$ of endpoint is called the {\it order} of the tree. If $s_v=1$ we say 
that $v$ is a {\it trivial} vertex. Finally if $s_v>1$ we say that $v$ is a 
{\it branching point} or {\it non-trivial} vertex. The tree structure 
induce a natural ordering (denoted by $<$) on the vertex such that if $v_1$ and 
$v_2$ are two vertices and $v_1<v_2$, then $h_{v_1}<h_{v_2}$. We call 
$\TT_{h,n}$ the set of all tree constructed in this way.

\item Given a vertex $v$ of $\t\in\TT_{h,n}$ that is not an endpoint, we can 
consider the subtrees of $\t$ with root $v$, which correspond to the connected 
components of the restriction of $\t$ to the vertices $w> v$. If a subtree 
with root $v$ contains only $v$ and an endpoint on scale $h_v+1$, we will 
call it a {\it trivial subtree}.

\item \label {i5} With each endpoint $v$ we associate one of the monomials 
contributing to\\ 
$\RR_1 {\cal V}^{(0)}(\psi^{(\le h_v-1)})$ and a set $\xx_v$ of space-time
points.

\item We introduce a {\it field label} $f$ to distinguish the field variables 
appearing in the terms associated with the endpoints described in item 
\ref{i5}); the set of field labels associated with the endpoint $v$ will be 
called $I_v$, $\xx(f)$, $\e(f)$ will be the position and type of the 
field variable $f$. Observe that $|I_v|$ is the order of the monomial 
contributing to 
${\cal V}^{(0)}(\psi^{(\le h_v-1)})$ and associated to $v$. Analogously, if $v$ 
is not an endpoint, we shall call $I_v$ the set of field labels associated with 
the endpoints following the vertex $v$; finally we will call the set of point 
$\xx(f)$ for $f\in I_v$ the {\it cluster} associated to $v$.

\end{enumerate}

Given ${\cal U}_i(\psi^{(h)})$ for $i=1,\ldots ,n$ we define the {\it truncated 
expectation} on scale $h$ as
\begin{align}
\EE^T_{h}&
\big[{\cal U}_1(\psi^{(h)});\ldots; {\cal U}_n^{(h)}(\psi^{( 
h)})\big]=\\
&\frac{\partial^n}{\partial\lambda_1\cdots\partial\lambda_n}\log\int 
P\left(d\psi^{(h)}\right)e^{\lambda_1{\cal 
U}_1(\psi^{(h)})+\cdots+\lambda_n{\cal 
U}_n(\psi^{(h)})}\biggr|_{\l_1=\ldots=\l_N=0}.\nn
\end{align}
In terms of above trees, the effective potential ${\cal V}^{(h)}$, $h\le -1$, 
can be written as
\be 
{\cal V}^{(h)}(\psi^{(\le h)}) + \b L \lis e_{h+1}=
\sum_{n=1}^\io\sum_{\t\in\TT_{h,n}}
{\cal V}^{(h)}(\t,\psi^{(\le h)})\;,\label{2.41}
\ee
where, if $v_0$ is the first vertex of $\t$ and $\t_1,\ldots,\t_s$ ($s=s_{v_0}$) 
are the subtrees of $\t$ with root $v_0$, ${\cal V}^{(h)}(\t,\psi^{(\le h)})$ is 
defined inductively as follows:
\begin{enumerate}[i]

\item  if $s>1$, then
\begin{align}
&{\cal V}^{(h)}(\t,\psi^{(\le h)})=\\
&\qquad\qquad\frac{(-1)^{s+1}}{ s!}\EE^T_{h+1}
\big[\bar{\cal V}^{(h+1)}(\t_1,\psi^{(\le h+1)});\ldots; \bar{\cal V}^{(h+1)}
(\t_{s},\psi^{(\le h+1)})\big]\;,\nn\label{2.42}
\end{align}
where $\bar{\cal V}^{(h+1)}(\t_i,\psi^{(\le h+1)})$ is equal to $\RR_1{\cal 
V}^{(h+1)}(\t_i,\psi^{(\le h+1)})$ if the subtree $\t_i$ contains more than one 
end-point, or if it contains one end-point but it is not a trivial subtree; it 
is equal to $\RR_1{\cal V}^{(0)}(\t_i,\psi^{(\le h+1)})$ if $\t_i$ is a trivial 
subtree;

\item if $s=1$ and $\t_1$ is 
not a trivial subtree, then ${\cal V}^{(h)}(\t,\psi^{(\le h)})$ is equal to
\[
\EE^T_{h+1}\big[\RR_1{\cal V}^{(h+1)}(\t_1,\psi^{(\le h+1)})\big].
\]
\end{enumerate}
Using its inductive definition, the right hand side of (\ref{2.41}) can be
further expanded, and in order to describe the resulting expansion we need some
more definitions.

We associate with any vertex $v$ of the tree a subset $P_v$ of $I_v$, the {\it 
external fields} of $v$. These subsets must satisfy various constraints. First 
of all, if $v$ is not an endpoint and $v_1,\ldots,v_{s_v}$ are the $s_v\ge 1$ 
vertices immediately following it, then $P_v \subseteq \cup_i P_{v_i}$; if $v$ 
is an endpoint, $P_v=I_v$. If $v$ is not an endpoint, we shall denote by 
$Q_{v_i}$ the intersection of $P_v$ and $P_{v_i}$; this definition implies that 
$P_v=\cup_i Q_{v_i}$. The union ${\cal I}_v$ of the subsets $P_{v_i}\setminus 
Q_{v_i}$ is, by definition, the set of the {\it internal fields} of $v$, and is 
non empty if $s_v>1$. Given $\t\in\TT_{h,n}$, there are many possible choices of 
the subsets $P_v$, $v\in\t$, compatible with all the constraints. We shall 
denote ${\cal P}_\t$ the family of all these choices and ${\bf P}$ the elements 
of ${\cal P}_\t$.

\def\PP{{\bf P}}

With these definitions, we can rewrite ${\cal V}^{(h)}(\t,\psi^{(\le h)})$ in 
the r.h.s. of (\ref{2.41}) as:
\bea &&{\cal V}^{(h)}(\t,\psi^{(\le
h)})=\sum_{{\bf P}\in{\cal P}_\t}
{\cal V}^{(h)}(\t,{\bf P})\;,\nn\\
&&{\cal V}^{(h)}(\t,{\bf P})=\int d\xx_{v_0}
\widetilde\psi^{(\le h)}(P_{v_0})
K_{\t,\PP}^{(h+1)}(\xx_{v_0})\;,\label{2.43}\eea
where
\be \widetilde\psi^{(\le h)}
(P_{v})=\prod_{f\in P_v}\psi^{ (\le
h)\e(f)}_{\xx(f)}\label{2.44}\ee
and $K_{\t,\PP}^{(h+1)}(\xx_{v_0})$ is defined inductively by the equation, 
valid for any $v\in\t$ which is not an endpoint,
\be K_{\t,\PP}^{(h_v)}(\xx_v)=\frac{1}{ s_v !}
\prod_{i=1}^{s_v} [K^{(h_v+1)}_{v_i}(\xx_{v_i})]\; \;\EE^T_{h_v}[
\widetilde\psi^{(h_v)}(P_{v_1}\setminus Q_{v_1}),\ldots,
\widetilde\psi^{(h_v)}(P_{v_{s_v}}\setminus
Q_{v_{s_v}})]\;,\label{2.45}\ee
where $\widetilde\psi^{(h_v)}(P_{v_i}\setminus Q_{v_i})$ has a definition 
similar to (\ref{2.44}). Moreover, if $v_i$ is an endpoint 
$K^{(h_v+1)}_{v_i}(\xx_{v_i})$ is equal to one of the kernels of the monomials 
contributing to $\RR_1{\cal V}^{(0)}(\psi^{(\le h_v)})$; if $v_i$ is not an 
endpoint, $K_{v_i}^{(h_v+1)}=K_{\t_i,\PP_i}^{(h_v+1)}$, where ${\bf P}_i=\{P_w, 
w\in\t_i\}$.

The final form of our expansions not yet given by (\ref{2.41})--(\ref{2.45}). 
We can further 
decompose ${\cal V}^{(h)}(\t,\PP)$, by using the following representation of the 
truncated expectation in the r.h.s. of (\ref{2.45}). Let us put $s=s_v$, 
$P_i\=P_{v_i}\setminus Q_{v_i}$; moreover we order in an arbitrary way the sets 
$P_i^\pm\=\{f\in P_i,\e(f)=\pm\}$, we call $f_{ij}^\pm$ their elements and we 
define $\xx^{(i)}=\cup_{f\in P_i^-}\xx(f)$, $\yy^{(i)}=\cup_{f\in P_i^+}\xx(f)$, 
$\xx_{ij}=\xx(f^-_{ij})$, $\yy_{ij}=\xx(f^+_{ij})$. Note that $\sum_{i=1}^s 
|P_i^-|=\sum_{i=1}^s |P_i^+|\=n$, otherwise the truncated expectation vanishes.

Then, we use the {\it Brydges-Battle-Federbush} \cite{B,GK,Le} formula saying 
that, up to a sign, if $s>1$,
\be \EE^T_{h}(\widetilde\psi^{(h)}(P_1),\ldots,
\widetilde\psi^{(h)}(P_s))=\sum_{T}\prod_{l\in T}
g^{(h)}(\xx_l-\yy_l)
\int dP_{T}({\bf t})\; {\rm det}\, G^{h,T}({\bf t})\;,\label{2.46}\ee
where $T$ is a set of lines forming an {\it anchored tree graph} between the 
clusters associated with $v_i$ 
that is $T$ is a set of lines, 
which becomes a tree graph if one identifies all the points in the same cluster. 
Moreover ${\bf t}=\{t_{ii'}\in [0,1], 1\le i,i' \le s\}$, $dP_{T}({\bf t})$ is a 
probability measure with support on a set of ${\bf t}$ such that $t_{ii'}={\bf 
u}_i\cdot{\bf u}_{i'}$ for some family of vectors ${\bf u}_i\in \RRR^s$ of unit 
norm. Finally $G^{h,T}({\bf t})$ is a $(n-s+1)\times (n-s+1)$ matrix, whose 
elements are given by
\be G^{h,T}_{ij,i'j'}=t_{ii'}g^{(h)}(\xx_{ij}-\yy_{i'j'})\;,
\label{2.48}\ee
with $(f^-_{ij}, f^+_{i'j'})$ not belonging to $T$. In the following we shall 
use (\ref{2.44}) even for $s=1$, when $T$ is empty, by interpreting the r.h.s. 
as equal to $1$, if $|P_1|=0$, otherwise as equal to ${\rm det}\,G^{h}= 
\EE^T_{h}(\widetilde\psi^{(h)}(P_1))$. It is crucial to note that $G^{h,T}$ is a 
Gram matrix, i.e., the matrix elements in (\ref{2.48}) can be written in terms 
of scalar products:
\bea
&& t_{ii'}
g^{(h)}(\xx_{ij}-\yy_{i'j'})
=\label{2.48a}\\
&&\hskip.3truecm
=\Big({\bf u}_i\otimes
A(\xx_{ij}-\cdot)\;,\ {\bf u}_{i'}\otimes B(\xx_{i'j'}-\cdot)\Big)\=({\bf f}_\a,{\bf g}_\b)
\nn
\;,
\eea
where
\bea
&&A(\xx)=\int d\kk
e^{-i\kk\xx}\sqrt{f_h(\kk)}
\frac{\sqrt{\left|\hat D_h(\kk)\right|}}{\hat D_h(\kk)}
\;,\\
&&B(\xx)=\int d\kk
e^{-i\kk\xx}\sqrt{f_h(\kk)}\frac{1}{\sqrt{\left|\hat D_h(\kk)\right|}}\;.\nn
\label{2.48b}
\eea
where $\hat D_h(\kk)=-i k_0(1+z_h)+(1+\a_h)\cos k-1+r+\g^{h}\m_h$.
The symbol $(\cdot,\cdot)$ denotes the inner product, i.e.,
\be
\big({\bf u}_i\otimes A(\xx-\cdot),
{\bf u}_{i'}\otimes B(\xx'-\cdot)\big)
=({\bf u}_i\cdot{\bf u}_{i'})\,\,\cdot\int d\zz A^*(\xx-\zz)B(\xx'-\zz)\;,
\label{2.48c}\ee
and the vectors ${\bf f}_\a,{\bf g}_\b$ with $\a,\b=1,\ldots,n-s+1$ are 
implicitly defined by (\ref{2.48a}). The usefulness of the representation 
(\ref{2.48a}) is that, by the Gram-Hadamard inequality, $|\det({\bf f}_\a,{\bf 
g}_\b)|\le \prod_{\a}||{\bf f}_\a||\,||{\bf g}_\a||$. In our case, $||{\bf 
f}_\a||,||{\bf g}_\a||\le C\g^{h/4}$ as it easily follows along the 
line of the proof of Lemma \ref{L1}. Therefore, $||{\bf f}_\a||\, 
||{\bf g}_\a||\le C \g^\frac{h}{ 2}$, uniformly in $\a$, so that the Gram 
determinant 
can be bounded by $C^{n-s+1}\g^{\frac{h}{ 2}(n-s+1)}$.
\\

If we apply the expansion (\ref{2.46}) in each vertex of $\t$ different from the 
endpoints, we get an expression of the form
\be {\cal V}^{(h)}(\t,\PP) = \sum_{T\in {\bf T}} \int d\xx_{v_0}
\widetilde\psi^{(\le h)}(P_{v_0}) W_{\t,\PP,T}^{(h)}(\xx_{v_0})
\= \sum_{T\in {\bf T}}
{\cal V}^{(h)}(\t,\PP,T)\;,\label{2.49}\ee
where ${\bf T}$ is a special family of graphs on the set of points $\xx_{v_0}$, 
obtained by putting together an anchored tree graph $T_v$ for each non trivial 
vertex $v$. Note that any graph $T\in {\bf T}$ becomes a tree graph on 
$\xx_{v_0}$, if one identifies all the points in the sets $\xx_v$, with $v$ an 
endpoint.

\subsection{Analyticity of the effective potentials.} 

Our next goal is the proof of the following result.\\

\begin{lemma}\label{L2} There exists a constants $\l_0>0$, independent of $\b$, 
$L$ and $r$, such that the kernels $W^{(h)}_{l}$ in the domain $|\l|\le \l_0$, 
are analytic function of $\l$ and satisfy for $h\ge h^*$
\be \frac1{\b L}\int d\xx_1\cdots d\xx_{l}|W^{(h)}_{l}
(\xx_1,\ldots,\xx_{l})|\le
\g^{h \left(\frac{3}{2}-\frac{l}{4}\right)}\g^\th
\,(C|\l|)^{max(1,l-1)}\label{2.52}\ee
with $\th=\frac{1}{4}$.
\end{lemma}
\vskip.3cm
\begin{proof} 
The proof is done by induction. We assume that for $k\ge h+1$
\pref{2.52} holds together with
\be
 \int d\xx(|x_0|+|x_1|^2) |W_2^{(k)}(\xx)|\le
C|\lambda|\g^{\th k}\label{as1}
\ee
and
\be
\int d\xx |W_2^{(k)}(\xx)|\leq 
C|\lambda| |r|\g^{\th k}.
\label{as2}
\ee
The validity of \pref{as1} and \pref{as2} implies \pref{b}. 

We now prove that the validity of \pref{2.52}, \pref{as1} and \pref{as2}.
Using the tree expansion described above and, in particular, (\ref{2.41}), 
(\ref{2.43}), (\ref{2.49}), we find that the l.h.s. of (\ref{2.52}) can be 
bounded above by
\begin{align} 
\sum_{n\ge 1}\sum_{\t\in {\cal T}_{h,n}}&
\sum_{\substack{\PP\in{\cal P}_\t\\ |P_{v_0}|=l}}\sum_{T\in{\bf T}}
C^n  \left[\prod_{i=1}^n
C^{p_i} |\l|^{\,\frac{p_i}2-1}\right]
\left[\prod_{v\ {\rm not}\ {\rm e.p.}} 
\frac{1}{s_v!}
\g^{{h_v}\left(\sum_{i=1}^{s_v}\frac{|P_{v_i}|}{ 4}-\frac{|P_v|}{ 
4}-\frac{3}{2}(s_v-1)\right)}\right]\cdot\label{2.57}\\
&\left[\prod_{v\ {\rm not}\ {\rm e.p.}}
\g^{-(h_v-h_{v'})z_1(P_v)}\right] 
\left[\prod_{v\ {\rm e.p.},
|I_v|=4,6} \g^{ h_{v'} \frac{|I_v|-2}{ 2 }} \right] \left[\prod_{v\
{\rm e.p.}, |I_v|=2} \g^\frac{ 3h_{v'}}{ 2} \right] \nn 
\end{align}
where $z_1(P_v)=2$ for $|P_v|=6$, $z_1(P_v)=1$ for $|P_v|=4$ and
$z_1(P_v)=\frac{3}{ 2}$ for $|P_v|=2$. 
Note the role of the $\RR_1$ operation in the above bound;
if we neglect $\RR_1$ we can get a similar bound where the second line of 
eq.\eqref{2.57} 
is simply replaced by $1$. Its proofs is an immediate consequence of 
the Gram--Hadamard inequality
\be |{\rm det} G^{h_v,T_v}({\bf t}_v)| \le
c^{\sum_{i=1}^{s_v}|P_{v_i}|-|P_v|-2(s_v-1)}\cdot\;
\g^{\frac{h_v}{ 2}
\left(\sum_{i=1}^{s_v}\frac{|P_{v_i}|}{ 2}-\frac{|P_v|}{ 
2}-(s_v-1)\right)}\;.\label{2.54}\ee
and of the decay properties of $g^{(h)}(\xx)$, implying
\be \prod_{v\ {\rm not}\ {\rm e.p.}}
\frac{1}{ s_v!}\int \prod_{l\in T_v} d(\xx_l-\yy_l)\,
||g^{(h_v)}(\xx_l-\yy_l)||\le c^n \prod_{v\ {\rm not}\ {\rm e.p.}}
\frac{1}{s_v!} \g^{-h_v(s_v-1)}\;.\label{2.55}\ee
If we take into account the subtraction to the 2 field terms and rewriting of 
the 4 and 6 fields terms involved in the $\RR_1$ operation we obtain the extra 
factor 
\[
\Big[\prod_{v\ {\rm not}\ {\rm e.p.}} \g^{-(h_v-h_{v'})z_1(P_v)}\Big] 
\Big[\prod_{v\ {\rm e.p.}, |I_v|=4,6} \g^{ h_{v'} \frac{|I_v|-2}{ 2 }} \Big] 
\Big[\prod_{v {\rm e.p.}, |I_v|=2} \g^\frac{ 3h_{v'}}{ 2} \Big]
\]
which is 
produced by the extra zeros and derivatives in the fields $D_{\xx_i,\xx_j}$ 
(when written as in the last of \pref{D}) and  $H_{\xx_1,\xx_2}$; each time or 
space derivative produce a gain $\g^{h_{v'}}$ or $\g^{h_{v'}/2}$ respectively 
while the zeros can be associated to the propagators in the anchored tree 
$T$ (for vertices that are not end points) or to the kernels in $V^{(0)}$ (for 
the end points) producing a loss bounded by $\g^{-h_{v}}$ or $\g^{-h_{v}/2}$. 
While the origin of such factors can be easily understood by the above 
dimensional considerations, some care has to be taken to obtain such 
gains, related to the presence of the interpolated points and to avoid "bad"  
extra factorials; we refer for instance to section 3 of \cite{BM1} where a 
similar 
bound in an analogous case is derived  with all details.

Once the bound \pref{2.57} is obtained, we have to see if we can sum over 
the scales and the trees. Let us define $n(v)=\sum_{i: v_i^*>v}\,1$ as the 
number of endpoints following $v$ on $\t$. 
Recalling that $|I_v|$ is the number of field labels 
associated to the endpoints following $v$ on $\t$ 
and using that
\begin{align}
& \sum_{v\ {\rm not}\ {\rm e.p.}}\left[\sum_{i=1}^{s_v}
|P_{v_i}|-|P_v|\right]=|I_{v_0}|-|P_{v_0}|\;,\nn\\
&\sum_{v\ {\rm not}\ {\rm e.p.}}(s_v-1)=n-1\;\,\label{56}\\
& \sum_{v\ {\rm not}\ {\rm e.p.}}
(h_v-h)\left[\sum_{i=1}^{s_v}
|P_{v_i}|-|P_v|\right]=\sum_{v\ {\rm not}\ {\rm e.p.}}
(h_v-h_{v'})(|I_v|-|P_v|)\;,\nn\\
&\sum_{v\ {\rm not}\ {\rm e.p.}}(h_v-h)(s_v-1)=
\sum_{v\ {\rm not}\ {\rm e.p.}}(h_v-h_{v'})(n(v)-1)\;,\nn
\end{align}
we find that (\ref{2.57}) can be bounded above by
\begin{align}
\sum_{n\ge 1}\sum_{\t\in {\cal T}_{h,n}}
\sum_{\substack{\PP\in{\cal P}_\t\\ |P_{v_0}|=2l}}\sum_{T\in{\bf T}}
C^n&  \g^{h\left(\frac{3}{ 2}-\frac{1}{ 4}|P_{v_0}|+\frac{1}{ 
4}|I_{v_0}|-\frac{3}{ 2}n\right)} \left[\prod_{i=1}^n C^{p_i} 
|\l|^{\,\frac{p_i}2-1}\right]\crcr
&\left[\prod_{v\ {\rm not}\ {\rm e.p.}} \frac{1}{s_v!}
\g^{(h_v-h_{v'})\left(\frac{3}{ 2}-\frac{|P_v|}{ 4}+\frac{|I_v|}{ 4}-\frac{3}{ 
2}n(v)+z_1(P_v)\right)}\right]\nn\\
&\left[\prod_{v\ {\rm e.p.}, |I_v|=4,6} \g^{ h_{v'} \frac{|I_v|-2}{
2 }} \right] \left[\prod_{v\ {\rm e.p.}, |I_v|=2} \g^\frac{3 h_{v'}}{
2} \right]
\end{align}
Using the identities
\begin{align} 
\g^{h  n}
\prod_{v\ {\rm not}\ {\rm e.p.}}
\g^{(h_v-h_{v'}
) n(v)}&=\prod_{v\ {\rm e.p.}}
\g^{h_{v'}}\;,\nn\\
\g^{h  |I_{v_0}|}
\prod_{v\ {\rm not}\ {\rm e.p.}}
\g^{(h_v-h_{v'}) |I_v|}&=\prod_{v\ {\rm e.p.}}
\g^{h_{v'} |I_v|}\;,\label{2.60}
\end{align}
we obtain
\begin{align}
\frac1{\b L}\int d\xx_1\cdots &d\xx_{l}|W^{(h)}_{l}
(\xx_1,\ldots,\xx_{l})|\le \crcr
&\sum_{n\ge 1}\sum_{\t\in {\cal
T}_{h,n}} \sum_{\substack{\PP\in{\cal P}_\t\\ |P_{v_0}|=2l}}\sum_{T\in{\bf
T}} C^n  \g^{h\left(\frac{3}{2}-\frac{l}{4}\right)} \left[\prod_{i=1}^n
C^{p_i} |\l|^{\,\frac{p_i}2-1}\right]\cdot\label{2.61}\\
&\cdot \left[\prod_{v\ {\rm not}\ {\rm e.p.}}
\frac{1}{s_v!} 
\g^{-(h_v-h_{v'})\left(\frac{|P_v|}{4}-\frac{3}{2}+z_1(P_v)\right)}\right]
\cdot\nn\\
&\cdot\left[\prod_{v\ {\rm e.p.}, |I_v|>6} 
\g^{h_{v'}\left(\frac{|I_v|}{4}-\frac{3}{2}\right)}
\right] \left[\prod_{v\ {\rm e.p.}, |I_v|=2} \g^{\frac{h_{v'}}{2}} \right]
\left[\prod_{v\ {\rm e.p.}, |I_v|=4,6} \g^{ h_{v'} \frac{3| I_v|-10}{4}}
\right] \nn 
\end{align}
Note that,
\be 
\left[\prod_{v\ {\rm e.p.}, |I_v|>6} 
\g^{h_{v'}\left(\frac{|I_v|}{4}-\frac{3}{2}\right)}
\right] \left[\prod_{v\ {\rm e.p.}, |I_v|=2} \g^{ \frac{h_{v'}}{2}} \right]
\left[\prod_{v\ {\rm e.p.}, |I_v|=4,6} \g^{ h_{v'} 
\frac{3|I_v|-10}{4}}\right]\le \g^{\frac{\bar h}{2}} \label{2.61a}\;,\ee
with $\bar h$ the highest scale label of the tree. Since
\be
\frac{|P_v|}{4}-\frac{3}{2}+z_1(P_v) \ge \frac{1}{2}\label{101}
\ee
we see that
\begin{align}
&\left[\prod_{v\ {\rm not}\ {\rm e.p.}} \frac{1}{s_v!}
\g^{-(h_v-h_{v'})\left(\frac{|P_v|}{4}-\frac{3}{2}+z_1(P_v)\right)}\right]
\g^{\frac{\bar h}{2}}\le\crcr
&\qquad\left[\prod_{v\ {\rm not}\ {\rm e.p.}}
\frac{1}{s_v!} 
\g^{-(h_v-h_{v'})\eta\left(\frac{|P_v|}{4}-\frac{3}{2}+z_1(P_v)\right)}
\right] \g^{h\frac{(1-\eta)}{2}}. 
\end{align}
for any $0<\eta<1$. On the other hand we have that
\be
 \frac{|P_v|}{4}-\frac{3}{2}+z_1(P_v) \ge \frac{|P_v|}{16}
\ee
so that, using also eq.\eqref{101}, we get
\be
\prod_{v\ {\rm not}\ {\rm e.p.}}
\frac{1}{s_v!} 
\g^{-(h_v-h_{v'})\eta\left(\frac{|P_v|}{4}-\frac{3}{2}+z_1(P_v)\right)}
\leq \left[\prod_{v\ {\rm not}\ {\rm e.p.}} \frac{1}{s_v!}
\g^{-\frac{\eta}{4}(h_v-h_{v'})}\right]\left[\prod_{v\ {\rm not}\ {\rm e.p.}}
\g^{-\frac{\eta}{32}|P_v|}\right]
\ee
Collecting the above estimates and using that the number of terms in 
$\sum_{T\in {\bf T}}$ is bounded by $C^n\prod_{v\ {\rm not}\ {\rm e.p.}} s_v!$ 
we obtain
\begin{align}
 \frac1{\b L}\int d\xx_1\cdots d\xx_{l}|&W^{(h)}_{l}
(\xx_1,\ldots,\xx_{l})|\le \g^{h\left(\frac{3}{2}-\frac{l}{4}\right)}\g^{ 
\frac{1-\eta}{2}h}
\sum_{n\ge 1}C^n\sum_{\t\in {\cal T}_{h,n}}\label{2.61b}\left[\prod_{i=1}^n
C^{p_i} |\l|^{\,\frac{p_i}2-1}\right]\cdot\crcr
&\cdot\left[\prod_{v\ {\rm not}\ {\rm e.p.}}
\g^{-(h_v-h_{v'})\frac{\eta}{4}}\right]
\sum_{\substack{\PP\in{\cal P}_\t\\|P_{v_0}|=2l}}\left[\prod_{v\ {\rm not}\ 
{\rm e.p.}}
\g^{-\frac{|P_v|}{64}}\right]\;.
\end{align}
\vskip.3cm
\0{\bf Remark}: eq. \pref{2.61b} says that a gain $\g^{\frac{\bar h}{2}}$ at 
the 
scale of the endpoint, see \pref{2.61a}, implies a gain 
$\g^{h\frac{1-\eta}{2}}$ at the 
root scale, as consequence of the fact that the renormalized scaling dimension 
of all vertices of the trees is strictly positive and $\ge \frac{1}{2}$; this 
property, which will be extensively  used below, is called {\it short memory 
property}. 
\vskip.3cm

The sum over $\PP$ can be bounded using the following combinatorial
inequality: let $\{p_v, v\in \t\}$,
with $\t\in\TT_{h,n}$, be a set of integers such that
$p_v\le \sum_{i=1}^{s_v} p_{v_i}$ for all $v\in\t$ which are not endpoints;
then, if $\a>0$,
$$\prod_{\rm v\;not\; e.p.} \sum_{p_v} \g^{-{\a p_v}}
\le C_\a^n\;.$$
This implies that
\[
\sum_{\substack{\PP\in{\cal P}_\t\\
|P_{v_0}|=2l}}\left[\prod_{v\ {\rm not}\ {\rm e.p.}}
\g^{-|P_v|\frac{\eta}{32}}\right]\prod_{i=1}^n
C^{p_i} |\l|^{\,\frac{p_i}2-1}\le C^n|\l|^n\;.
\]
Finally
$$\sum_{\t\in {\cal T}_{h,n}}
\prod_{v\ {\rm not}\ {\rm e.p.}}
\g^{-(h_v-h_{v'})\frac{\eta}{4}}\le C^n\;,$$
as it follows from the fact that the number of non trivial vertices in $\t$
is smaller than $n-1$ and that the number of trees in ${\cal T}_{h,n}$ is
bounded by ${\rm const}^n$. Altogether we obtain
\be 
\frac1{\b L}\int d\xx_1\cdots d\xx_{l}|W^{(h)}_{l}
(\xx_1,\ldots,\xx_{l})|\le 
\g^{h\left(\frac{3}{2}-\frac{l}{4}\right)}\g^{\th h} \sum_{n\ge
1}C^n |\l|^n\;,\label{2.61e}
\ee
where we have set $\th=(1-\eta)/2$. Moreover we choose $\eta=\frac12$ so that 
$\th=\frac14$.
Once convergence is established, the limit $L,\b\to\io$ is a straightforward 
consequence, see for instance section 2 of \cite{BFM1}.

In order to complete the proof we need to show the validity of the inductive 
assumption \pref{as1}\pref{as2}. It is clearly true for $h=1$; moreover, 
by the bound \pref{2.61e} we get \pref{as1}. We have finally to prove \pref{as2}.
We can write $g^{(h)}(\xx)=g^{(h)}|_{r=0}(\xx)+r^{(h)}(\xx)$ where 
$g^{(h)}|_{r=0}$ is the single scale propagator of the $r=0$ case and $r^{(h)}$ 
satisfkies
\[
 \left|\partial_0^{n_0} \tilde\partial_1^{n_1}\tilde g^{(h)}(\xx)\right|\le
C_N\frac{|r|\g^{-\frac{h}{2}}}{1+[\g^h|x_0|+\g^{\frac{h}{2}}|x|]^N}\g^{
h(n_0+n_1/2)}
\]
that is the same bound \pref{b} with an extra $\frac{|r|}{\g^h}$. We can 
therefore 
write
\be
\hat W^{(h)}_2(0)=\hat W^{(h)}_{2,a}(0)+\hat W^{(h)}_{2,b}(0)
\ee
where $\hat W^{(h)}_{2,a}(0)$ is the effective potential of the $r=0$ case. We 
will 
show below that $\sum_{h=-\io}^1 \hat W^{(h)}_{2,a}(0)=0$ and as a consequence 
$|\sum_{h=k}^1 \hat W^{(h)}_{2,a}(0)|\le C|\l|\g^{(1+\th)k}$ as 
$|\hat W^{(h)}_{2,a}(0)|\le C|\l|\g^{(1+\th)h}$. On the other hand 
$|\hat W^{(h)}_{2,b}(0)|\le C|\l| |r| \g^{\th h}$ so that 
\be
\g^{h-1}\m_{h-1}=\g^h\m_h+\hat W^{(h)}_{2}(0)
\ee
hence $\g^{h-1}\m_{h-1}=\sum_{h=h}^1\hat W^{(k)}_{2}(0)$ and
$|\m_h|\le C|\l|$. 

It remains to prove that $\sum_{h=-\io}^1 \hat W^{(h)}_{2,a}(0)=0$.
This can be checked noting that in the $r=0$ case it is more natural to consider
the following ultraviolet regularization, instead of \pref{pro}
at $\b=\io$
\be
g(\xx)=\th_M(x_0)\int_{-\pi}^{\pi} dk e^{i k x+\e(k) x_0} 
\ee
with $\e(k)=\cos(k)-1$ and $\th_M(x_0)$ is a smooth function with support in 
$(\g^{-M},+\io)$; note that $g(\xx)$ verifies \pref{pro1}. We can write 
$g(\xx)=g^{(u.v.)}(\xx)+g^{(i.r.)}(\xx)$ with $g^{(u.v.)}(\xx)=h(x_0) g(\xx)$ 
and $g^{(i.r.)}(\xx)=(1-h(x_0)) g(\xx)$, and $h(x_0)$ a smooth function $=1$ is 
$|x_0|<1$ and $=0$ if $|x_0|>\g$. The integration of the ultraviolet part can be 
done as in section 3 of \cite{BGPS}, writing $\th_M(x_0)$ as sum of compact 
support 
functions. After that, the limit $M\to\io$ can be taken, and we can write 
$g^{(i.r.)}(\xx)=\sum_{h=-\io}^{-1} g^{(h)}(\xx)$ with 
\be
g^{(h)}(\xx)=\th(x_0)(1-h(x_0))\int_{-\pi}^{\pi} dk c_h(k)e^{i k x+\e(k) x_0} 
\ee
with $c_h(k)$ a smooth function non vanishing for $\pi\g^{h-1}\le |k|\le \pi 
\g^{h+1}$; note that $g^{(h)}(\xx)$ verifies \pref{b}, and the integration of 
the infrared scales is essentially identical to the one described in this 
sections. Once all scales are integrated out, we obtain kernels 
$W^{(-\io)}_{n,m}$ coinciding with the ones obtained before; however with this 
choice of the ultraviolet cut-off, $W^{(-\io)}_{2,0}\equiv 0$ is 
an immediate consequence of the presence of the $\th_M(x_0)$ in the propagator. 
Indeed the kernels can be written as sum over Feynman graphs, which contain 
surely a closed fermionic loop or a tadpole (the interaction is local in time).
\end{proof}

\subsection{The 2-point Schwinger function in the insulating phase.}

In the case 
$r=0$ we have $h^*=-\io$ and the integration considered in this section 
conclude the construction of the effective potential. 
Similarly, if $r<0$ and $|\l|$ is small then $g^{(< h^*)}\equiv 0$, so that 
again the construction of the effective potential is concluded by the 
integration on scale $h^*$.

In both case the analysis described above can be easily extended 
to take into account the external fields, that is $\phi\not=0$
(see for instance section 3.4 of \cite{BM1} for details in a similar case).
The 2-point Schwinger function can be written as, if we define 
$h_\kk={\rm min}\{h: \hat g(\kk)\not=0\}$
\be
\hat S(\kk)=\sum_{j=h_\kk}^{h_\kk+1} Q^{(j)}_\kk \hat g^{(j)}(\kk)
Q^{(j)}_\kk-\sum_{j=h_\kk}^{h_\kk+1} G^{(j)}(\kk)\hat W^{(j-1)}_{2,0}(\kk)G^{(j)}(\kk)
\ee 
where from \pref{2.52} $|\hat W^{(j-1)}_{2,0}(\kk)|\le C|\l|\g^{(1+\th)j}$, 
$Q^{(h)}_\kk$ is defined inductively by the relation $Q^{(1)}=1$ and %
\be
Q^{(h)}_\kk=1-\hat W^{(h)}_{2,0}(\kk) g^{(h+1)}(\kk) Q^{(h+1)}_\kk
\ee
and
\be
G^{(h+1)}(\kk)=\sum_{k=h=1}^1 g^{(k)}(\kk) Q^{(k)}_\kk
\ee
so that
\be
|Q^{(h)}_\kk-1|\le C|\l|\g^{\th h}\quad |\hat G^{(j)}(\kk)|\le 
C|\l|\g^{-\th h}
\ee
Using that
\be
\hat g^{(h)}(\kk)=\frac{f_h(\kk)}
{-i k_0(1+z_{-\io})+(1+\a_{-\io})\frac{k^2}{2}}+\hat r^{(h)}(\kk)
\ee
where
\be
|\hat r^{(h)}(\kk)|\le C\g^{(1-\th)h}
\ee
so that \pref{mama} follows.

\section{Renormalization Group integration: the second regime in the metallic 
phase.}

We have now to consider the integration of the scales with $h< h^*$, that is
\be
e^{-\WW(0)}=e^{-\beta L F_{h^*}}\int P(d\psi^{(\le h^*)})e^{V^{(h^*)}(\psi) 
}\label{v16}
\ee
where $P(d\psi^{(\le h^*)})$ has propagator given by
\be
 g^{(\le h^*)}(\xx)=\int d\kk\frac{\chi_{\le h^*}(\kk)}{ -i 
k_0(1+z_{h^*})+(1+\a_{h^*})(\cos k-1)+r+\g^{h^*}\m_{h^*} }\label{24a}
\ee
with $z_{h^*},\a_{h^*},\n_{h^*}=O(\l)$.

The denominator of the propagator \eqref{24a} vanishes in correspondence of the
two Fermi momenta and we need a multiscale decomposition.  It is convenient to 
rewrite \pref{v16} in the following way
\be
\int P(d\psi^{(\le h^*)})e^{V^{(h^*)}(\psi)} = \int \tilde P(d\psi^{(\le 
h^*)})e^{V^{(h^*)}(\psi)-\g^{h^*}
\n_{h^*}\int d\xx\psi^+_{\xx}\psi^-_{\xx} }\label{v16p}
\ee
where $\tilde P(d\psi^{(\le h^*)})$ has propagator
\be
g^{(\le h^*)}(\xx)=\int d\kk e^{i\kk\xx}\frac{\chi_{\le h^*}(\kk)}{ -i
k_0(1+z_{h^*})+(1+\a_{h^*})(\cos k-\cos p_F) }\label{24ap}
\ee
and 
\be
(1+\a_{h^*})\cos 
p_F=(1+\a_{h^*})-r-\g^{h^*}\m_{h^*}+\g^{h^*}\n_{h^*}\label{cccc}
\ee
Observe that, assuming that also $\nu_{h^*}\leq K|\l|$, then we have 
\be
 C_-\sqrt r\leq p_F\leq C_+ \sqrt r \label{pf}
\ee
for $\l$ small enough.
The strategy of the analysis is the following: 
\begin{enumerate}[a)]
 \item we will perform a multiscale analysis of \eqref{v16}. In this analysis 
we will have to chose $\nu_{h^*}=O(\l)$ as function of $p_F$ and $\l$ to obatin 
a convergent expansion. 

\item at the end of the above construction we will use
\pref{cccc} to obtain the Fermi momentum $p_F$ as function of $\l$ and $r$. 
\end{enumerate}

We can now write
\[
 \chi_{\leq h^*}(\kk)=\chi_{\leq h^*,1}(\kk)+\chi_{\leq h^*,-1}(\kk)
\]
where
\[
 \chi_{\leq h^*,\o}(\kk)=\tilde\vartheta\left(\o \frac{k}{p_F}\right) 
\chi_{\leq h^*}(\kk)
\]
where $\o=\pm 1$, $\tilde\vartheta$ is a smooth function such that 
$\tilde\vartheta(k)=1$ 
for $k>\frac{1}{2}$ and $\tilde\vartheta(k)=0$ for $k<-\frac{1}{2}$ and
\[
 \tilde\vartheta(k)+\tilde\vartheta(-k)=1
\]
for every $k$. Thus 
$\tilde\vartheta\left(\frac{k}{p_F}\right)$ is  equal to 1 in a neighbor of 
$p_F$ and 0 in a neighbor of $-p_F$. Clearly  $\chi_{\leq 
h^*,\pm 1}(\kk)$ is a smooth, compact support function and it allows us to write
\be\label{omega}
 g^{(\leq h^*)}(\xx)=\sum_{\o=\pm 1} e^{i\o p_F x} g^{(\leq h^*)}_\o(\xx)
\ee
where
\be
g^{(\leq h^*)}_\o(\xx)=\int d\kk
e^{i(\kk-\o\pp_F)\xx}\frac{\chi_{\leq h^*,\o}(\kk)}{ -i 
k_0(1+z_{h^*})+(1+\a_{h^*})(\cos
k-\cos p_F) }
\ee
with $\pp_F=(0,p_F)$. 

We observe that, if the running coupling constants were not present in the 
cut-off function $\chi_{\leq h}$, we could have used as a quasi-particle 
cut-off function
\[
\tilde\chi_{\leq h^*,\pm 1}(\kk)=\vartheta\left(\pm k\right) 
\chi_{\leq h^*}(\kk)
\]
where $\vartheta(k)=1$ if $k>0$ and 
$\vartheta(k)=0$ if $k<0$. Indeed, thanks to \eqref{cha}, this would have made 
essentially no difference. On the other hand, thanks to \eqref{pf} and 
\eqref{hh}, we have that $\chi_{\leq h^*,\pm 1}$ differs from 
$\tilde\chi_{\leq h^*,\pm 1}$ only for a finite number (not depending on $r$) 
of scales so that this does not modify our qualitative picture. Finally notice 
that the argument of 
$\tilde\vartheta$ is not scaled with $\g^{-h}$ but only with 
$p_F^{-1}=O(\g^{-\frac{h^*}2})$. 

The multiscale integration  is done exactly as in \cite{BM1}. The localization operation
is defined in the following way
\begin{align}
\LL_2 \int d\underline\xx  
W_{4}(\xx_1,\xx_2,\xx_3,\xx_4)\prod_{i=1}^4\psi^{\e_i}_{\xx_i,\o_i}=&
\hat W_4(0)\int d\xx 
\psi^+_{\xx,1}\psi^{-}_{\xx,1}\psi^+_{\xx,-1}\psi^{-}_{\xx,-1}\nn\\
\LL_2 \int d\underline\xx 
W_{2}(\xx_1,\xx_2)\psi^{+}_{\xx_1,\o}\psi^{-}_{\xx_2,\o}=&
\hat W_2(0)\int\psi^{+}_{\xx,\o}\psi^{-}_{\xx,\o}d\xx+\nn\\
\partial_1 \hat W_2(0)
\int\bar 
\psi^{+}_{\xx,\o}\Delta_1 \psi^{-}_{\xx,\o}d\xx+
&\partial_0\hat W_2(0)\int 
\psi^{+}_{\xx,\o}\partial_0 \psi^{-}_{\xx,\o}d\xx
\end{align}
where 
\[
 \bar\Delta_1 f(\xx)=2\int d\kk (\cos k-\cos p_F) e^{i\kk\xx}\hat 
f(\kk)\quad \hbox{if} \quad f(\xx)=\int d\kk e^{i\kk\xx}\hat f(\kk)
\]
Note that in the kernels $W_l$ are included the oscillating factors $e^{i\o p_F 
x}$ coming form \eqref{omega}.

After the integration of the scale $\psi^{h^*},..\psi^h$ we get
\be
e^{-\WW(0)}=e^{-\beta L F_{h}}\int P_{Z_h}(d\psi^{(\le 
h)})e^{V^h(\sqrt{Z_h}\psi)}
\ee
where $P_{Z_h}(d\psi^{(\le h)})$ is the Grasmann integration with propagator 
$\frac{g_\o^{(\leq h)}}{Z_h}$ where
\be\label{pan}
 g_\o^{(\leq h)}(\xx)=\int d\kk e^{i(\kk-\o\pp_F)} \frac{\chi_{\leq 
h,\o}(\kk)}{(1+z_{h^*})ik_0+(1+\alpha_{h^*})(\cos k-\cos p_F)}
\ee
We can now write
\[
 \int P_{Z_h}(d\psi^{(\le h)})e^{V^h(\sqrt{Z_h}\psi)}=\int 
\tilde P_{Z_{h-1}}(d\psi^{(\le h)})e^{\tilde V^h(\sqrt{Z_h}\psi)}
\]
where
\begin{align}
\LL_2 \tilde V^h=&l_h \int d\xx 
\psi^+_{\xx,1}\psi^{-}_{1,\xx}\psi^+_{\xx,-1}\psi^{-}_{-1,\xx}+(a_h-z_h)
\sum_{\o}\int d\xx\psi^+_{\o,\xx}\partial\psi_{\o,\xx}+\crcr 
&n_h\int d\xx\psi^+_{\o,\xx}
\psi_{\o,\xx}
\end{align}
while $\tilde P_h(d\psi^{(\le h)})$ is the integration with propagator 
identical to \eqref{pan} but with $\chi_{\leq h,\o}(\kk)$ replaced by 
$ \frac{\chi_{\leq h,\o}(\kk)}{\tilde Z_{h-1}(\kk)}$ with
\[
\tilde Z_{h-1}(\kk)=Z_h+\chi_{\leq h,\o}(\kk)Z_h z_h
\]
Setting $Z_{h-1}=\tilde Z_{h-1}(0)$, we can finally write
\be
\int P_{Z_h}(d\psi^{(\le h)})e^{V^h(\sqrt{Z_h}\psi)}=e^{-\b L e_h}\int 
P_{Z_{h-1}}(d\psi^{(\le h-1)})\int \tilde P_{Z_{h-1}}(d\psi^{(h)})
 e^{-\bar V^{(h)}(\psi^{(\le h)})}
\ee
where $\tilde P_{Z_{h-1}}$ is the integration with propagator $\frac{\tilde 
g_\o^{(h)}}{Z_{h-1}}$
\[
\tilde g_\o^{(h)}(\xx)=\int d\kk e^{i\xx(\kk-\o\pp_F)}\frac{\tilde 
f_{h,\o}(\kk)} 
{(1+z_{h^*})ik_0+(1+\alpha_{h^*})(\cos k-\cos p_F)}
\]
where
\[
 \tilde f_{h,\o}(\kk) =Z_{h-1}\left[\frac{\chi_{\leq 
h,\o}(\kk)}{\tilde Z_{h-1}(\kk)}- \frac{\chi_{\leq h-1,\o}(\kk)}{Z_{h-1}}\right].
\]
Finally we have
\[
\bar V^{(h)}(\psi^{(\le h)})=\tilde  
V^{(h)}\left(\sqrt{\frac{Z_h}{Z_{h-1}}}\psi^{(\le h)}\right)
\]
so that
\begin{align}
\LL_2 \tilde V^h=\l_h \int d\xx 
\psi^+_{\xx,1}\psi^{-}_{1,\xx}\psi^+_{\xx,-1}\psi^{-}_{-1,\xx}+\d_h
\sum_{\o}\int d\xx\psi^+_{\o,\xx}\partial\psi_{\o,\xx}+\nn\\
\g^h \n_h\sum_{\o}\int 
d\xx\psi^+_{\o,\xx}
\psi_{\o,\xx}\label{loc1}
\end{align}
with
\[
 \g^h \nu_h=\frac{Z_h}{Z_{h-1}}n_h\qquad 
\d_h=\frac{Z_h}{Z_{h-1}}(a_h-z_h)\qquad 
\l_h=\left(\frac{Z_h}{Z_{h-1}}\right)^2l_h
\]
We can now prove the following:

\begin{lemma}\label{L3}
For $h\leq h^*$, every $N$ and $\l$ small enough we have

\be\label{pp2}
|\partial_0^{n_0} \partial_1^{n_1}\tilde g^{(h)}_\o(\xx)|\le
C_N\frac{v_F^{-1}\g^h}{1+[\g^h|x_0|+v_F^{-1}\g^h|x|]^N}\g^{
h(n_0+n_1)}v_F^{-n_1} 
\ee
with $v_F=\sin(p_F)=O(r^{\frac{1}{2}})$.
\end{lemma}

\begin{proof}
We can write
\[
 \cos k-\cos p_F=\cos p_F(\cos(k-\o p_F)-1)+\o v_F\sin(k-\o p_F)
\]
Using \eqref{hhe} it easily follows that
\[
|\sin(k-p_F)|\leq C\g^hr^{-\frac{1}{2}}
\]
and thus $|k-p_F|\leq C\g^hr^{-\frac{1}{2}}$. For this case we get
\be\label{support1}
 \int \tilde\vartheta(k)\chi_{\leq h}(\kk) d\kk \leq C r^{-\frac{1}{2}}\g^{2h}.
\ee
The analogous of \eqref{esti} for the present Lemma is
\be\label{esti1}
|x_0^{N_0}x^{N_1}\tilde 
g^{(h)}_\o(\xx)|\le C
\g^{h(1-N_0-N_1)}r^{\frac{N_1-1}{2}}
\ee
To prove it we observe that \eqref{deriv} and \eqref{scaling} remain true. 
Indeed the only difference arise due to the presence of $\tilde\vartheta$. But 
this does not change the estimates since its derivative gives a smaller factor 
$O(\g^{-\frac{h^*}2})$ as compared to the factor $O(\sqrt{r}\g^{-h})$ 
coming from the derivative 
of $\tilde f_{h,\o}$.
Again we use \eqref{prods} for $P_1\le N_1/2$ together with
\[
\g^{-h(1+N_0+P_1)}\leq \g^{-h(1+N_0+N_1)}r^{N_1-P_1}\leq 
\g^{-h(1+N_0+N_1)}r^{\frac{N_1}2}.
\]
Reasoning as before, for $P_1>N_1/2$ we get
\be
\prod_{i=i}^{P_1}\frac{d^{p_1}}{dk^{p_i}}\sin(k)\le C
\g^{(2P_1-N_1)h}r^{\frac{N_1}{2}-P_1}.
\ee
Observing that
\[
\g^{-h(1+N_0+P_1)}\g^{(2P_1-N_1)h}r^{\frac{N_1}{2}-P_1}=\g^{-h(1+N_0+N_1)}r^{
\frac{N_1}{2}}\left(\frac{\g^h}{r}\right)^{P_1}\leq C\g^{-h(1+N_0+N_1)}r^{
\frac{N_1}{2}}
\]
and collecting we get
\be
|\partial_0^{N_0}\partial_1^{N_1} \hat 
g^{(h)}(\kk)|\le \g^{-(1+N_0+N_1)h}r^{\frac{N_1}{2}}
\ee
The Lemma follows easily combining the above estimate with \eqref{support1} and 
the analogous of \eqref{powerk}.
\hfill\end{proof}

Again the effective potential can be written as a sum over trees similar to the 
previous ones but with the following modifications:
\begin{enumerate}
\item We associate a label $h\le h^*$ with the root.

\item With each endpoint $v$ we associate one of the monomials
contributing to \\ $\RR_2 {\cal V}^{(h^*)}(\psi^{(\le h_v-1)})$
or one of the terms contributing to $\LL_2\VV^{(h_v)}
(\psi^{(\le h_v-1)})$.
\end{enumerate}

The main result of this section is the following Lemma.

\begin{lemma} Assume that
\be
|\l_k|,|\d_k|\le C v_F |\l|\quad  |\n_k|\le C |\l|\label{indp}
\ee
than there exists a constants $\l_0>0$, independent
of $\b$, $L$ and $r$, such that, for $h<h^*$, the kernels $W^{(h)}_{l}$
are analytic functions of  $\lambda$ for $|\l|\le \l_0$. Moreover they
satisfy
\be \frac1{\b L}\int d\xx_1\cdots d\xx_{l}|W^{(h)}_{l}
(\xx_1,\ldots,\xx_l)|\le
\g^{h \left(2-\frac{l}{2}\right)} v_F^{\frac{l}{2}-1}
 \,(C|\l|)^{max(1,l-1)}\;.\label{2.52a}\ee
\end{lemma}
\vskip.3cm
\begin{proof}
The proof of this Lemma follows closely the line of \cite{BM1}. The only major 
difference is the presence of the small factors in \eqref{pp2}. We will report 
only the modification of the proof needed to deal with those factors.

We start noting that the analogous of the bound \pref{2.57}
becomes
\begin{align}  \frac1{\b L}&\int d\xx_1\cdots d\xx_{l}|W^{(h)}_{l}
(\xx_1,\ldots,\xx_l)|\le\nn\\
&\sum_{n\ge 1}\sum_{\t\in {\cal T}_{h,n}}
\sum_{\substack{\PP\in{\cal P}_\t\\|P_{v_0}|=
l}}\sum_{T\in{\bf T}}
C^n \left[\prod_{v\ {\rm not}\ {\rm e.p.}} \frac{1}{s_v!}
\g^{h_v\left(\sum_{i=1}^{s_v}\frac{|P_{v_i}|}{ 2}-\frac{|P_v|}{
2}-2(s_v-1)\right)}\right]\crcr
&\left[\prod_{v\ {\rm not}\ {\rm 
e.p.}}\left(\frac{1}{v_F}\right)^{\sum_{i=1}^{s_v}\frac{|P_{v_i}|}{ 
2}-\frac{|P_v|}{
2}-(s_v-1)}\right]\left[\prod_{v\ {\rm not}\ {\rm e.p.}}
\g^{-(h_v-h_{v'})z_2(P_v)}\right]\\
&\left[\prod_{v\ {\rm e.p.}; v\in I^{R}, |I_v|\ge 6} 
|\l|\g^{h^*\left(\frac{3}{
2}-\frac{|I_v|}{ 4}\right)} \right] 
\left[\prod_{v\ {\rm e.p.}; v\in I^{R}, |I_v|=2,4} |\l|\g^{h^*\left(\frac{3}{
2}-\frac{|I_v|}{ 4}\right)+z_2(P_v) (h_{v'}- h^*)} \right] \nn\\
&\left[\prod_{v\ {\rm e.p.}; v\in I^{\l}} |\l| v_F\right] 
\left[\prod_{v\ {\rm e.p.}; v\in I^{\n,\d}} |\l|  \g^{h_{v'}}\right]
\left[\prod_{i=1}^n C^{p_i}\right]\nn
\label{2.57a}
\end{align}
where:
\begin{enumerate}[1.]
\item
the last factor keeps into account the presence of the factors $Z_h/Z_{h-1}$;
\item
the factor $\left(\frac{1}{v_F}\right)^{\frac{|P_{v_i}|}{ 2}-\frac{|P_v|}{
2}-(s_v-1)}$ comes from the bound on the Gram determinant and the fact that  
$|g^h(\xx)|\le \frac{\g^h}{v_F}$; 

\item $z_2(P_v)=1$ for $|P_v|=4$ and $z_2(p_v)=2$ for $|P_v|=2$;

\item $I^R$ is the set of endpoints associated to $\RR\VV^{(h^*)}$ and the 
factor  $\g^{h^*\left(\frac{3}{2}-\frac{|I_v|}{4}\right)}$ comes from the 
bound \pref{2.52}; 

\item $I^\l$ is the set of end-points associated to $\l_k$ and the factor $v_F$ 
comes from \pref{indp};  

\item $I^\d$ is the set of end-points associated to $\d_k$ and 
the derivative in  \pref{loc1} produces an extra $\g^{h_{v'}}/v_F$;

\item $I^\n$ is the set of end-points associated to $\n_k$ and the factor 
$\g^{h_{v'}}$ comes from \pref{loc1}.

\end{enumerate}

Proceeding like in the proof of Lemma \ref{L2} using \pref{56} we get
\begin{align}
\frac1{\b L}&\int d\xx_1\cdots d\xx_{l}|W^{(h)}_{l}
(\xx_1,\ldots,\xx_l)|\le\nn\\
&\sum_{n\ge 1}\sum_{\t\in {\cal T}_{h,n}}
\sum_{\substack{\PP\in{\cal P}_\t\\ |P_{v_0}|=l}}\sum_{T\in{\bf T}}
C^n  \g^{h\left(2-\frac{1}{2}|P_{v_0}|+\frac{1}{2}|I_{v_0}|-2 n\right)}\crcr
&\left[\prod_{v\ {\rm not}\ {\rm e.p.}} \frac{1}{s_v!}
\g^{(h_v-h_{v'})\left(2-\frac{|P_v|}{ 2}+\frac{|I_v|}{2}- 2 
n(v)+z_2(P_v)\right)}\right]\nn\\
&\left[\prod_{v\ {\rm e.p.}; v\in I^{R}} 
|\l|\g^{h_{v'}\left(\frac{3}{2}-\frac{|I_v|}{4}\right)} \right] 
\left[\prod_{v\ {\rm e.p.}; v\in I^{\l}} |\l| v_F\right] 
\left[\prod_{v\ {\rm e.p.}; v\in I^{\n,\d}} |\l|  \g^{h_{v'}}\right]
\left[\prod_{i=1}^n C^{p_i}\right]\nn\\
&\left[\prod_{v\ {\rm not}\ {\rm e.p.}}
\left(\frac{1}{v_F}\right)^{\left(\sum_{i=1}^{s_v}\frac{|P_{v_i}|}{2}-\frac{
|P_v| } { 2 } -(s_v-1)\right)}\right]
\end{align}

Finally using \pref{2.60} we arrive to
\begin{align} \frac1{\b L}&\int d\xx_1\cdots d\xx_{l}|W^{(h)}_{l}
(\xx_1,\ldots,\xx_l)|\le\nn\\
&\sum_{n\ge 1}\sum_{\t\in {\cal T}_{h,n}}
\sum_{\substack{\PP\in{\cal P}_\t\\ |P_{v_0}|=l}}\sum_{T\in{\bf T}}
C^n  \g^{h\left(2-\frac{1}{2}|P_{v_0}|\right)}
\left[\prod_{v\ {\rm not}\ {\rm e.p.}} \frac{1}{s_v!}
\g^{(h_v-h_{v'})\left(2-\frac{|P_v|}{ 2}+z_2(P_v)\right)}\right]\nn\\
&\left[\prod_{v\ {\rm e.p.}, v\in I^R} 
\g^{h_{v'}\left(-\frac{1}{2}+\frac{|I_v|}{4}\right)}\right]
\left[\prod_{v\ {\rm e.p.}; v\in I^{\l}} |\l| v_F\right] 
\left[\prod_{v\ {\rm e.p.}; v\in I^{\n,\d}} |\l|\right]
\left[\prod_{i=1}^n C^{p_i}\right]\nn\\
&\left[\prod_{v\ {\rm not}\ {\rm e.p.}}
\left(\frac{1}{v_F}\right)^{\left(\sum_{i=1}^{s_v}\frac{|P_{v_i}|}{2}-
\frac{|P_v|}{2}-(s_v-1)\right)}\right]
\end{align}
Because $\g^{h'_v}\le \g^{h^*}\le v_F^2$ and $|I_v|\ge 2$ we have
\[
 \g^{h_{v'}\left(-\frac{1}{2}+\frac{|I_v|}{4}\right)}\leq
v_F^{-1+\frac{|I_v|}{2}}
\]
so that
\begin{align}
\frac1{\b L}\int &d\xx_1\cdots d\xx_{l}|W^{(h)}_{l}
(\xx_1,\ldots,\xx_l)|\le\nn\\
&\sum_{n\ge 1}\sum_{\t\in {\cal T}_{h,n}}
\sum_{\substack{\PP\in{\cal P}_\t\\ |P_{v_0}|=l}}\sum_{T\in{\bf T}}
C^n  \g^{h\left(2-\frac{1}{2}|P_{v_0}|\right)}
\left[\prod_{v\ {\rm not}\ {\rm e.p.}} \frac{1}{s_v!}
\g^{(h_v-h_{v'})\left(2-\frac{|P_v|}{ 2}+z_2(P_v)\right)}\right]\nn\\
&\left[\prod_{v\ {\rm e.p.}, v\in I^R, I^\l} 
v_F^{-1+\frac{|I_v|}{2}}\right]\left[\prod_{v\ {\rm not}\ {\rm e.p.}}
\left(\frac{1}{v_F}\right)^{\left(\sum_{i=1}^{s_v}\frac{|P_{v_i}|}{2}-
\frac{|P_v|}{2}-(s_v-1)\right)}\right]
|\l|^n\left[\prod_{i=1}^n C^{p_i}\right]\label{82}
\end{align}
For $v\in I^\d,I^\n$, one has $|I_v|=2$ so that $v_F^{-1+\frac{|I_v|}{ 2}}=1$,  
and we can write
\be
\left[\prod_{v\ {\rm e.p.}, v\in I^R, I^\l} 
v_F^{-1+\frac{|I_v|}{2}}\right]=\left[\prod_{v\ {\rm e.p.}} 
v_F^{-1+\frac{|I_v|}{2}}\right]= 
v_F^{-n+\sum_{v\ {\rm e.p} }\frac{|I_v|}{2}}
\ee
Using that
\[
 \sum_v (s_v-1)=n-1\qquad\qquad \sum_{v\ {\rm e.p.} 
}|I_v|=l+\sum_v\sum_{i=1}^{s_v}(|P_{v_i}|-|P_v|)
\]
we get
\begin{align}
\prod_{v\ {\rm e.p.}} v_F^{-1}
\prod_{v\ {\rm not}\ {\rm e.p.}} 
\left(\frac{1}{v_F}\right)^{-(s_v-1)}=& v_F^{-1}\nonumber\\
\prod_{v\ {\rm e.p.}} v_F^{\frac{|I_v|}{2}}
\prod_{v\ {\rm not}\ {\rm e.p.}} 
\left(\frac{1}{v_F}\right)^{\left(\sum_{i=1}^{s_v}\frac{|P_{v_i}|}{2}-
\frac{|P_v|}{2}\right)}=& v_F^{\frac{l}{2}}
\end{align}
Collecting these estimates we get 
\begin{align}
\frac1{\b L}\int& d\xx_1\cdots d\xx_{l}|W^{(h)}_{l}
(\xx_1,\ldots,\xx_l)|\le\\
&v_F^{\frac{l}{2}-1}\sum_{n\ge 1}\sum_{\t\in {\cal T}_{h,n}}
\sum_{\substack{\PP\in{\cal P}_\t\\ |P_{v_0}|=l}}\sum_{T\in{\bf T}}
C^n  \g^{h\left(2-\frac{1}{2}|P_{v_0}|\right)}
\left[\prod_{v\ {\rm not}\ {\rm e.p.}} \frac{1}{s_v!}
\g^{(h_v-h_{v'})\left(2-\frac{|P_v|}{2}+z_2(P_v)\right)}\right]\crcr
&|\l|^n\left[\prod_{i=1}^n C^{p_i}\right]\nn
\end{align}
Performing the sums as in the previous section we 
prove \pref{2.52a}.\hfill\end{proof}

\vskip.3cm
\0{\bf Remarks.} 
\begin{itemize}
\item Observe that, for $h\ge h^*$, bound \pref{2.52} says that the $L_1$ norm 
of the effective potential is $O(\g^{h (3/2-l/4)})$ while, for $h\le h^*$, bound 
\pref{2.52a} says that the $L_1$ norm of the effective potential is $O(\g^{h 
(2-l/2)} v_F^{\frac{l}{2}-1})$; the two bounds coincide of course at 
$h\simeq h^*$ since $\g^{h^*}\sim r$, $v_F\sim \sqrt{r}$ so that $r^{(2-l/2)} 
r^{\frac{l}{4}-\frac{1}{2}}=r^{\frac{3}{2}-\frac{l}{4}}$.

\item The fact that the Fermi velocity vanishes as $r$ approaches 0
produces the "dangerous" factor 
$\left(\frac{1}{v_F}\right)^{\frac{|P_{v_i}|}{2}-\frac{|P_v|}{2}-(s_v-1)}$
in \pref{2.57a} which is diverging as $r\to 0$. This is compensated by the 
extra factors of $v_F$ associated to the difference between the 
scaling dimensions the first and second regime, that is
\be
\left[\prod_{v\ {\rm
e.p.}} \g^{h_{v'}\left(\frac{3}{2}-\frac{|I_v|}{4}\right)}\right]=
\left[\prod_{v\ {\rm
e.p.}} \g^{h_{v'}\left(2-\frac{|I_v|}{2}\right)} \right]\left[\prod_{v\ {\rm
e.p.}} \g^{h_{v'}\left(-\frac{1}{2}+\frac{|I_v|}{4}\right)} \right]
\ee
\end{itemize}

\subsection{The flow of the running coupling constants}

We now prove by induction that, for $h\le h^*$ and $\th=\frac{1}{4}$ we have
\be
|\l_h|\le C |\l| r^{\frac{1}{2}+\th},\quad |\d_h|\le C |\l| r^{\frac{1}{2}+\th} 
\quad |\n_h|\le C|\l| \g^{\th h}\label{indpp}
\ee
First we check that \pref{indp} is true for $h=h^*$. By definition of the 
$\LL_2$ operation
\be
\l_{h^*}=\l[\hat v(0)-\hat v(2 p_F)]+O(\l^2 \g^{h^*(\frac{1}{2}+\th)})
\ee
where the second term in the r.h.s comes from \pref{2.52}; as $\hat v(k)$ is 
even the first term is $O(r)$ so that surely $\l_{h^*}$ vanishes as 
$O\left(r^{\frac{1}{2}+\th}\right)$. Moreover from  \pref{2.52}, taking into 
account that a 
derivative $\partial_1$ gives an extra $\g^{-h/2}$, that is
\be\label{gh}
\int d\xx |\partial_1 W_2^{(h^*)}(\xx)|\leq 
C|\lambda|\g^{h^*\left(\frac{1}{2}+\th\right)}
\ee
we get
\be
|\d_{h^*}|\le C\g^{h\left(\frac{1}{2}+\th\right)} |\l| \le C|\l| 
r^{\frac{1}{2}+\th} 
\ee

The flow of $\n_h$ is given by
\be
\n_{h-1}=\g\n_h+\b^{(h)}_\n(\vec v_h,...,\vec v_0)
\ee
where $\vec v_h=(\l_h,\d_h,\n_h)$. We can decompose the propagator as
\be\label{dec}
\tilde g^{(h)}_\o(\xx)=g^{(h)}_{\o,L}(\xx)+r^{(h)}_{\o}(\xx)
\ee
where
\be\label{lut}
g^{(h)}_{\o,L}(\xx)=\int d\kk e^{i\kk\xx}\frac{\tilde f_h(\kk)}{-i k_0+\o v_F k}
\ee
and $\tilde f_h$ has support contained in $C\g^{h-1}\le \sqrt{k_0^2+v_F^2 
k^2}\le C\g^{h+1}$. Moreover, for every $N$, we have
\be
|r^{(h)}_{\o}(\xx)|\le \left(\frac{\g^h}{v_F}\right)^3\frac{C_N}{1+\g^h 
(|x_0|+v_F^{-1}|x|)^N}\label{ssp}
\ee
that is the bound for $r^{(h)}_{\o}(\xx)$ has an extra factor $\g^{2h}/v_F^2\le 
\g^h$ with respect to the bound \eqref{pp2} for $\tilde g^{(h)}_\o(\xx)$. 

In the expansion for $\b^{(h)}_\n$ studied in the previous subsection, we can 
decompose every propagator as in \eqref{dec} and collect all the term that 
contains only $g_{L,\o}^{(h)}$ and that come from trees with no end-points 
associated to $\RR\VV^{(h^*)}$; this sum vanish due to parity. Therefore 
$\b^{(h)}_\n=O(\l \g^{\th h})$ and by iteration
\be
\n_{h-1}=\g^{-h+h^*}[\n_{h^*}+\sum_{k=h}^{h^*}\g^{k-h^*}\b^{(k)}_\n].
\ee
Thus we can choose $\n_{h^*}$ so that
\be
\n_{h^*}=-\sum_{k=-\io}^{h^*}\g^{k-h^*}\b^{(k)}_\n
\ee
This implies that
\be
\n_{h-1}=\g^{-h+h^*}[-\sum_{k=-\io}^{h}\g^{k-h^*}\b^{(k)}_\n]
\ee
and $|\n_h|\le C|\l|\g^{\th h}$.

We now study the flow equations for $\l_h$ and $\d_h$ with $h<h^*$ 
\bea
&&\l_{h-1}=\l_h+\b^{(h)}_\l(\vec v_h,...,\vec 
v_0)\nn\\&&\d_{h-1}=\d_h+\b^{(h)}_\d(\vec v_h,...,\vec v_0) 
\eea
where we have redefined $\d_0$  as to include the sum $\tilde \delta_0$ of the 
terms $O(\l)$, which satisfies
\be
|\tilde\delta_0|\leq C\left|\int d\kk k \partial^2 v(\kk+(\o-\o')\pp_F) 
g_\o^{\le h^*}(\kk)\right|
\ee
where one derivative over $v$ comes from the $\RR_1$ operation and the other 
from the definition of $\d$. Observe that
\be
|\tilde\delta_0|\le 
C\sum_{k\le h^*} v_F^{-2}\g^{2h}
\le C |\l| r
\ee
since $v_F k\le C\g^h$ in the support of $f_h$.

Again we can use \eqref{dec} and decompose the beta function
for $\a=\l,\d$ as
\be
\b^{(h)}_\a(\vec v_h,...,\vec v_0)=\bar\b^{(h)}_\a(\l_h,\d_h,
...,\l_0,\d_0)+\b^{(h)}_{\a,R}(\vec v_h,...,\vec v_0)
\ee
where $\bar\b^{(h)}_\a$ contains only propagators $g^{(h)}_{\o,L}(\xx)$ and 
end-points to which is associated $\l_k,\d_k$. Therefore $\b^{(h)}_{\a,R}$ 
contains either a propagator $r^{(h)}_{\o,L}(\xx)$, a $\n_k$ or an irrelevant 
term. Observe that

\begin{enumerate}
\item Terms containing a propagator $r_h$ or a factor $\n_h$ have an extra 
$\g^{\th h}$ in their bounds, therefore by an argument similar to the one used 
in (70) ({\it short memory property}) they can be bounded as 
$O(v_F \g^{\th h})$. The factor $v_F$ comes from the factor $v_F^{l/2-1}$ in 
\pref{2.52a} when $\a=\l$, and from the derivative $\partial_1$ in the 
case $\a=\d$.

\item The terms containing an irrelevant end-points associated to a term $\RR 
V^{(h^*)}$ have an extra $\g^{\th h^*}$ (coming from (59)) and an extra $\g^{\th 
(h-h^*)}$ for the short memory property; therefore they can be bound as $O(v_F 
\g^{\th h})$. The origin of the factor $v_F$ is the same as in the previous 
point. 

\end{enumerate}

In conclusion
\be
|\b^{(h)}_{\a,R}|\le C  v_F   \l^2
 \g^{\th h}
\label{xxb}
\ee
From \eqref{lut} it is easy to see that
\begin{align}
\bar\b_\l^{(h)}(\l_h,\d_h,
...,\l_0,\d_0)=&v_F \hat\b_\l^{(h)}\left(\frac{\l_h}{v_F},\frac{\d_h}{v_F},
...,\frac{\l_0}{v_F},\frac{\d_0}{ v_F}\right)\nn\\
\bar\b_\d^{(h)}(\l_h,\d_h,
...,\l_0,\d_0)=&v_F \hat\b_\d^{(h)}\left(\frac{\l_h}{v_F},\frac{\d_h}{v_F},
..,\frac{\l_0}{v_F},\frac{\d_0}{v_F}\right)
\end{align}
where  $\hat\b_\l^{(h)}(\l_d,\d_h,..,\l_,\d_0)$ is the beta function of a 
Luttinger model with $v_F=1$. It has been proved in \cite{BM2} that
\be
|\hat\b_\l^{(h)}(\l_d,\d_h,..,\l_0,\d_0)|\le C [\max(|\l_k|,|\d_k|)]^2 \g^{\th (h-h^*)}
\ee
therefore assuming by induction that $|\l_k|,|\d_k|\le 2 |\l| 
r^{\frac{1}{2}+\th}$ for $k\ge h$ we get
\be
|\bar\b_\a^{(h)}(\l_h,\d_h,
...,\l_0,\d_0)|\le 4 C v_F \l^2 r^{1+2\th}\g^{\th h} v_F^{-2}r^{-\th}
\le 4 
C v_F\l^2 \g^{\th h} r^\th\label{xxa}.
\ee
Thus
\be
|\l_{h-1}|\le |\l_{h^*}|+\sum_{k=h}^{h^*}4 
C v_F\l^2 \g^{\th h} r^\th\le 2|\l| r^{\frac{1}{2}+2 \th}
\ee
and the same is true for $\d_h$.

Moreover we have
\be
\frac{Z_{h-1}}{Z_h}=1+\b_z^{(h)}
\ee
so that
\be
\g^\h=1+\b^{-\io}\left(\frac{\l_{-\io}}{v_F}\right)
\ee
where $\b^{-\io}$ is the beta function with $v_F=1$; therefore
\be
Z_h=\g^{\h(h-h^*)}(1+A(\l))
\ee
with $|A(\l)|\le C|\l|$. Observe that $\h=O(\l^2 r^{4 \th})$, hence is 
vanishing as $r\to 0$ as $O(\l^2 r)$. 

Finally the inversion problem for $p_F$ can be studied as in section 2.9 of 
\cite{BFM1}. The analysis for the Schwinger function is done in a way 
similar to the one in section 3 above.

\end{document}